\documentclass[11pt]{article}

\usepackage{amsmath,amssymb,amsthm,hyperref,color}
\usepackage[margin = 1in]{geometry}
\usepackage{bbm}
\usepackage{bm}
\usepackage{verbatim}
\usepackage{graphicx}
\usepackage{caption}
\usepackage{tikz}

\hypersetup{colorlinks=true,citecolor=blue, linkcolor=blue, urlcolor=blue}

\newtheorem{lemma}{Lemma}[section]
\newtheorem{theorem}[lemma]{Theorem}

\newtheorem{fact}[lemma]{Fact}

\newtheorem{corollary}[lemma]{Corollary}

\theoremstyle{remark}
\newtheorem{remark}{Remark}
\theoremstyle{definition}
\newtheorem{definition}[lemma]{Definition}

\DeclareMathOperator*{\E}{\mathrm{E}}

\DeclareMathOperator*{\1}{\mathbbm{1}}

\DeclareMathOperator*{\boundary}{{\partial\!}}

\newcommand{\aA}{\alpha}
\newcommand{\bB}{\beta}
\newcommand{\gG}{\gamma}

\newcommand{\tT}{\theta}
\newcommand{\sS}{\sigma}

\newcommand{\DD}{\Delta}
\newcommand{\OO}{\Omega}
\newcommand{\TT}{\Theta}

\newcommand{\bl}{\backslash}

\newcommand{\pr}{\Pr}

\newcommand{\edge}[2]{\{#1,#2\}}

\newcommand{\eps}{\varepsilon}

\newcommand{\poly}{\mathsf{poly}}

\newcommand{\N}{\mathbb{N}}
\newcommand{\R}{\mathbb{R}}

\DeclareMathOperator{\e}{e}

\newcommand{\TV}[1]{\left\lVert #1 \right\rVert_{\textsc{tv}}}

\title{Structure Learning of $H$-colorings}

\author{Antonio Blanca \thanks{
  		School of Computer Science, Georgia Institute of Technology, Atlanta, GA 30332.
  		Research supported in part by NSF grants CCF-1617306 and CCF-1563838. 
  		Email: {\tt \{ablanca3,chenzongchen,vigoda\}@gatech.edu}}
  \and
  Zongchen Chen \footnotemark[1] 
  \and
  Eric Vigoda \footnotemark[1] 
  \and
	Daniel \v{S}tefankovi\v{c} \thanks{
   	Department of Computer Science, University of Rochester, Rochester, NY 14627.
   	Research supported in part by NSF grant CCF-1318374. 
   	Email: {\tt stefanko@cs.rochester.edu}
   }}

\begin{document}

\maketitle

\begin{abstract}
We study the structure learning problem for 
$H$-colorings, an important class of Markov random fields 
that
capture key combinatorial structures on graphs,
including proper colorings and independent sets,
as well as spin systems from statistical physics.
The learning problem is as follows: for a fixed (and known) constraint graph $H$ with $q$ colors 
and an unknown graph $G=(V,E)$ with $n$ vertices, given uniformly random $H$-colorings of $G$, 
how many samples are required to learn the edges of the unknown graph~$G$?
We give a characterization of $H$ 
for which the problem is identifiable for every~$G$, i.e., 
we can learn $G$ with an infinite number of samples.
We also show that there are identifiable constraint graphs 
for which
one cannot hope to learn every graph~$G$ efficiently.

We focus particular attention on the case of
proper vertex $q$-colorings of graphs of maximum degree~$d$ 
where intriguing connections to statistical physics
phase transitions appear.  
We prove that in the tree uniqueness region (i.e., when $q>d$)  the problem
is identifiable and we can learn $G$ in $\poly(d,q)\times O(n^2\log{n})$ time.
In contrast for soft-constraint systems, such as the Ising model, the best possible
running time is exponential in $d$.  In the tree non-uniqueness region (i.e., when $q\leq d$) we prove that
the problem is not identifiable and
thus $G$ cannot be learned.
Moreover, when $q<d-\sqrt{d} + \Theta(1)$ we prove that even learning an equivalent
graph (any graph with the same set of $H$-colorings) is computationally hard---sample complexity is exponential in $n$ in the worst case.
We further explore the connection between the efficiency/hardness of the structure learning problem and the uniqueness/non-uniqueness phase transition
for general $H$-colorings and prove that
under a well-known condition in statistical physics, known as
the Dobrushin uniqueness condition,  
we can learn~$G$ in $\poly(d,q)\times O(n^2\log{n})$ time.  
\end{abstract}


\section{Introduction}\label{sec:intro}



Structure learning is a
general framework
for supervised learning
where, instead of learning labels or real numbers as in classification or regression,
the task is to learn a more complex structure such as a graph.
A myriad of
fundamental learning problems 
can be studied
in this framework. 
Notably,
structure learning for Markov random fields (undirected graphical models), 
where the goal is to recover the underlying graph from random samples,
has found important applications in diverse fields, 
including the study of phylogeny \cite{Huel}, gene expression \cite{Mar}, protein interactions \cite{Mor}, neuroscience \cite{Sch}, image processing \cite{Roth} and sociology \cite{Eagle}.

Our goal in this paper is to understand when is
structure learning for Markov random fields
possible in polynomial time.
We mostly focus on the task of \textit{exact recovery}, where a learning algorithm is said to succeed only when it outputs exactly the hidden graph.
In applications, exact recovery of graphical 
models is often of interest since the true graph structure contains
valuable information about the dependencies in the model.
Consequently, it has been very well-studied; see, e.g.,~\cite{CL,Dasgupta,Srebro,CT,LGK,AHHK,JRVS,RWL,BMS,BGSa,Bresler,VMLC,HKM,KM}. 
While the typical setting in these works are \textit{soft-constraint models},
our focus here are models with \textit{hard constraints}.
Specifically, we consider the structure learning problem and the closely related question
of statistical identifiability in the general setting of $H$-colorings,
an important class of Markov
random fields that include all hard-constraint models
and
capture key combinatorial structures on graphs,
including proper colorings and independent sets.

Given an undirected, connected constraint graph $H=(V(H),E(H))$, 
with vertices $V(H) \!= \{1,\dots,q\}$ referred to as colors (or spins),
an $H$-coloring of a graph $G=(V,E)$
is an assignment of colors $\{1,\dots,q\}$ to the vertices of $G$
such that adjacent vertices of $G$ receive adjacent colors in $H$. That is, 
an $H$-coloring $\sigma$ is a mapping $\sigma:V\rightarrow V(H)$ such that if $\{v,w\}\in E$, then 
$\{\sigma(v),\sigma(w)\}\in E(H)$.
If such an assignment is possible we say that $G$ is $H$-colorable.
The constraint graph $H$ is allowed to have self-loops, but not parallel edges, and 
every $\edge{i}{j}$ such that 
$\edge{i}{j} \not\in E(H)$ 
is called a \textit{hard constraint}.

When, for example, $H$ is the complete graph on $q$ vertices with no self-loops, denoted $K_q$, 
neighboring vertices of $G$ must be assigned different colors, and 
thus the proper $q$-colorings of $G$ are precisely its $H$-colorings.
If, on the other hand,~$H$ is the graph with
two vertices $V(H)=\{0,1\}$  
and two edges $E(H) = \{\{0,0\},\{0,1\}\}$, 
then the subset of vertices
assigned color $1$ in any $H$-coloring of $G$ form an independent set.
Hence, in this case,
there is one-to-one correspondence between the independent sets and the $H$-colorings of $G$.
Spin systems without hard constraints, i.e., \textit{soft-constraint systems}, correspond to the
constraint graph with all possible edges; that is, $H=K^+_q$ which is
the complete graph $K_q$ with a self-loop at every vertex. In this case,
all $q^{|V|}$ labelings of $G$ are valid $H$-colorings. 

We consider structure learning and statistical identifiability for $H$-colorings with at least one hard constraint;
that is, $H \neq K_q^+$.
(Note that the missing edge could be a self-loop.)
$H$-colorings are well-studied in several other contexts.
For the decision problem (for a fixed $H$ is a graph $H$-colorable?), a dichotomy result~\cite{HNa,HNb,Bulatov05,Sigge} has been established characterizing for which $H$ the problem is either in P or is NP-complete. 
The corresponding dichotomy conjecture for 
the directed case (directed graph homomorphisms) has also received considerable 
attention~\cite{FV}; see also the recent works~\cite{Bulatov17,Zhuk} and the references therein.
The complexity of the exact counting version of the problem 
was characterized by Dyer and Greenhill~\cite{DG}, and 
the complexity of approximate counting/sampling was studied in~\cite{GKP,DGJ,GGJ}.

For an $H$-colorable graph $G$, 
let $\Omega_G^H$ be the set of all possible $H$-colorings of $G$ 
and let $\pi_G^H$ denote the uniform measure over $\Omega_G^H$. 
(Typically the constraint graph $H$ will be fixed and thus we will drop the dependence on $H$
in our notation.)
Some of our results 
extend to the more general setting of a weighted constraint graph $H$ and a weighted graph $G$,
where $\pi_G^H$
is the corresponding Gibbs distribution; see Section \ref{sec:weighted} for a precise definition. 

For statistical identifiability our goal is to
characterize the cases when every graph 
is learnable with an infinite number of samples.
For structure learning our goal is to efficiently learn the graph $G$ from samples drawn independently from 
$\pi_G$. More formally, let $H$ be a fixed constraint graph and
let $\mathcal{G}$ be a family of $H$-colorable graphs.
Suppose that
we are given $L$ samples $\sigma^{(1)},\sigma^{(2)},\ldots,\sigma^{(L)}$ drawn independently from the distribution $\pi_G^H$ where $G \in \mathcal{G}$. 
A \textit{structure learning algorithm} for the constraint graph $H$ and the graph family $\mathcal{G}$ 
takes as input the sample sequence $\sigma^{(1)},\sigma^{(2)},\ldots,\sigma^{(L)}$
and outputs an estimator $\hat{G} \in \mathcal{G}$ such that $\Pr[G=\hat{G}] \ge 1-\varepsilon$ where $\varepsilon > 0$ is a prescribed error (failure probability). 

As mentioned earlier,
structure learning has been well-studied 
for soft-constraint models where $H=K_q^+$. In the context of the Ising model, the most well-known and widely studied soft-constraint system, for an unknown graph $G$ with $n$ vertices, maximum degree $d$
and maximum interaction strength $\beta$ (which corresponds to the inverse temperature in the
homogenous model), 
Bresler \cite{Bresler} presented an algorithm to learn $G$ in $O(n^2\log{n})\times\exp(\exp(O(\beta d^2)))$ time.
A different algorithm was provided by Vuffray et al.\ \cite{VMLC}
with running time $O(n^4\log{n})\times\exp(O(\beta d))$.
Recently, Bresler's algorithm was extended to arbitrary Markov random fields 
by Hamilton et al.\ \cite{HKM}, and a new approach was presented by Klivans and Meka~\cite{KM}
which achieves nearly-optimal running time of $O(n^2\log{n})\times\exp(O(\beta d))$.
Both of these general results \cite{HKM,KM} are for the case of soft-constraint models
and do not apply to the setting of hard-constraint systems. 
We shall see that, while the algorithm in \cite{KM} achieves (optimal) single exponential dependence on~$d$
for general soft-constraint systems, 
the structure learning problem for
hard-constraint systems 
is quite different. 
Indeed,
some hard-constraint systems are not statistically identifiable (and thus the unknown graph $G$ cannot be learned);
others allow very efficient structure learning algorithms with $\poly(n,d,q)$ running time;
while in others any structure learning algorithm requires exponentially (in $n$) many samples.

For hard-constraint systems,
the structure learning problem was previously studied by 
Bresler, Gamarnik, and Shah~\cite{BGS-hc} for independent sets (more generally, for the hard-core model where 
the independent sets are weighted by their size and a model parameter $\lambda>0$).
They achieve nearly-optimal running time of $O(n^2\log{n})\times\exp(O(d\lambda))$.
For our positive results we generalize the structure learning algorithm in~\cite{BGS-hc}.

Finally, we remark that while our results aim to learn the underlying graph $G$ exactly (i.e., exact recovery), in some cases, we also consider the problem of learning an equivalent graph $G'$ such that $\pi_G = \pi_{G'}$ (equivalent-structure learning).
The corresponding approximation problem of finding a graph $G'$ such that $\pi_{G'}$ is close to $\pi_G$ in 
some notion of distance, such as total variation distance or Kullback-Leibler divergence (see, e.g., \cite{Abbeel,bresler2016learning}),
is
apparently much simpler for hard-constraint systems; see Section \ref{section:approx}.

\subsection{Results}

We first address the
statistical identifiability problem for
general $H$-colorings. 

\begin{definition}\label{def:identify}	
	A constraint graph $H$ is said to be identifiable with respect to a family of $H$-colorable graphs $\mathcal{G}$ if for any two distinct
	graphs
	$G_1,G_2\in\mathcal{G}$ we have $\pi_{G_1} \neq \pi_{G_2}$ (or equivalently $\OO_{G_1} \neq \OO_{G_2}$). In particular, when $\mathcal{G}$ is the set of all finite $H$-colorable graphs we say that $H$ is identifiable.	
\end{definition}
\noindent
To characterize identifiability we introduce a
supergraph $G_{ij}$ of $H$.
This supergraph will not be used as a constraint graph, but
rather we will consider the $H$-colorings of $G_{ij}$ in our characterization theorem.
Consider an edge $\{i,j\}\in E(H)$.  We construct $G_{ij}$ by starting from $H$
and duplicating the colors $i$ and $j$.
These new copies, denoted $i'$ and $j'$,
have the same neighbors as the original colors $i$ and $j$, respectively,
except for the one edge $\{i',j'\}$ which is not included; see Figure~\ref{fig:Hij} for an illustration of this supergraph 
and Definition \ref{dfn:Hij} for a formal definition. 

For a constraint graph $H$ we say that a pair of colors $i,j\in V(H)$ are compatible (resp., incompatible)
if $\{i,j\}\in E(H)$ (resp., $\{i,j\} \not\in E(H)$).  
Also, if $\sigma$ is an $H$-coloring of a graph $G=(V,E)$, we use $\sigma(v)$
to denote the color of $v \in V$ in $\sigma$.
Our characterization theorem considers whether in every $H$-coloring
of $G_{ij}$ the new vertices $i',j'$ receive compatible colors. 

\begin{theorem}\label{thm:identify}  
	Let $H\neq K_q^+$ be an arbitrary constraint graph. 
	If $H$ has at least one self-loop, then $H$ is identifiable.
	Otherwise,
	$H$ is identifiable if and only if for each $\edge{i}{j} \in E(H)$ there exists an $H$-coloring $\sigma$ of $G_{ij}$ such that
	$\sigma(i')$ and $\sigma(j')$ 
	are incompatible colors in $H$.
\end{theorem}

\noindent
The intuition for the role of the graph~$G_{ij}$ in our characterization theorem is that if we added the edge $\{i',j'\}$ to~$G_{ij}$ to 
form the graph $G'_{ij}=G_{ij}\cup\{i',j'\}$, then in every $H$-coloring of $G'_{ij}$ 
the vertices $i',j'$ 
receive compatible colors. If the same property holds for $G_{ij}$, i.e.,
in every $H$-coloring of $G_{ij}$ the vertices $i',j'$ receive compatible colors,
then the edge $\{i',j'\}$ plays no role and
the pair of graphs $G_{ij}$ and $G'_{ij}$ have the same set of $H$-colorings.
That is, $G_{ij}$ and $G'_{ij}$ are indistinguishable, 
and so $H$ is not identifiable.

Let $\mathcal{G}(n,d)$ be the family of $n$-vertex graphs of maximum degree at most $d$.
Our next result shows that
for some identifiable (with respect to $\mathcal{G}(n,d)$) constraint graphs, 
one cannot hope to learn the underlying graph $G$ efficiently.
We remark that this is not the case
for soft-constraint models, where one can always learn $G$ in
time $O(n^2\log{n})\times\exp(O(d\beta))$~\cite{KM}. 

\begin{theorem}
	\label{thm:neg-identify}
	There exists an identifiable constraint graph $H$ and a constant $c>0$ 
	such that, for all $n \ge 8$,  
	any structure learning algorithm for $H$ and the graph family $\mathcal{G}(n,7)$ 
	that succeeds with probability at least $\exp(-cn)$ 
	requires at least $\exp(cn)$ samples. 
\end{theorem}
\noindent
While this theorem shows that there is no efficient learning algorithm for all identifiable models,
for some relevant models
structure learning can be done efficiently.

We focus first on the
case of proper $q$-colorings where $H=K_q$.
In general, the colorings problem is not identifiable. However, if we consider identifiability with respect to the graph family $\mathcal{G}(n,d)$ we get a 
richer picture.
We prove that 
when $q>d$ the $q$-colorings problem is identifiable, whereas
when $q\leq d$ it is not.
This phase
transition for identifiability/non-identifiability at $q=d$ is quite intriguing since it coincides
with the statistical physics uniqueness/non-uniqueness phase transition of the Gibbs distribution on
infinite $d$-regular trees.

Uniqueness on an infinite $d$-regular tree $T$
may be defined as follows. 
Let $T_\ell$ be a finite $d$-regular tree of height $\ell$ rooted at vertex $r$ and
consider the uniform measure $\mu_\ell$ over the proper $q$-colorings of $T_\ell$.
We say that the Gibbs measure on $T$ is unique if, as $\ell \rightarrow \infty$, the configuration in the leaves 
of $T_\ell$
has no influence on 
the color of $r$. 
Formally, the Gibbs measure on $T$ is unique iff for every color $c \in \{1,\dots,q\}$
and every configuration $\tau$ on the leaves of $T_\ell$,
we have $\mu_\ell(r = c | \tau) \rightarrow 1/q$ as $\ell \rightarrow \infty$. That is, the conditional distribution at the root given $\tau$
is the uniform measure over $\{1,\dots,q\}$ as $\ell \rightarrow \infty$.
In~\cite{Jonasson,BW} it was shown that there is a unique Gibbs measure iff $q > d$.
The uniqueness/non-uniqueness phase transition on
the infinite $d$-regular tree 
is known to be closely connected to the efficiency/hardness 
of other fundamental computational problems, such as sampling and counting; see, e.g., \cite{Weitz,Sly,SlySun,GSV,LLY}.

In the identifiable region $q>d$ we present an efficient structure learning algorithm
with $O(n^2\log{n})\times\poly(d,q)$ running time.
When $q\leq d$, where we cannot hope to learn the hidden graph $G$ since there are pairs of graphs with the same set of $H$-colorings,  
we may be interested in a learning algorithm that outputs a graph $G'$ that is equivalent  to the unknown graph $G$, in the sense that $\Omega_{G'} = \Omega_G$. 
We say that an algorithm is an \textit{equivalent-structure learning algorithm} for a fixed constraint
graph $H$ and a graph family $\mathcal{G}$ if for every $G\in\mathcal{G}$,
with probability at least $1-\eps$,
the algorithm outputs $G' \in \mathcal{G}$ such that $\Omega_{G}=\Omega_{G'}$.

In the $q$-coloring setting,
it turns out that when $q \le d-\sqrt{d} + \Theta(1)$
there is a family of exponentially many $n$-vertex graphs 
with different sets of $q$-colorings 
that
only differ on an exponentially small
(in $n$) fraction of their $q$-colorings.
Consequently, any equivalent-structure learning algorithm
requires exponential many samples to distinguish among these graphs. 

Our results for statistical identifiability, structure learning and equivalent-structure learning for proper $q$-colorings are stated in the following theorem. 
\begin{theorem}
	\label{thm:col}
	Consider the $q$-colorings problem $H=K_q$.  The following hold for all $d$.
	\begin{enumerate}
		\item {\textbf{Efficient learning for $q>d$:}}
		\label{part:1}
		For all $q>d$, $n \ge 1$ and any $G \in \mathcal{G}(n,d)$, 
		there is a structure learning algorithm
		that given 	$L = O (qd^3 \log (\frac{n^2}{\varepsilon}))$ independent samples from $\pi_G$
		outputs 
		$G$
		with probability at least $1-\varepsilon$  
		and has running time $O(L n^2)$.
		\item {\textbf{Non-identifiability for $q\leq d$:}}
		\label{part:2}
		For all $q \leq d$ and $n \ge q+2$
		there exist $q$-colorable graphs $G_1, G_2 \in \mathcal{G}(n,d)$ such that $G_1 \neq G_2$ and
		$\pi_{G_1}=\pi_{G_2}$.
		
		\item 	{\textbf{Lower bound for $q <d-\sqrt{d} + \Theta(1)$:}} 
		\label{part:col-exp}
		For all $q <d-\sqrt{d} + \Theta(1)$ 
		there exists a constant $c>0$ 
		such that
		any equivalent-structure learning algorithm 
		for the family of graphs $\mathcal{G}(n,d)$ 
		that
		succeeds with probability at least $\exp(-cn)$ 
		requires at least $\exp(cn)$ samples, provided $n$ is sufficiently large.
		
	\end{enumerate}
\end{theorem}

\noindent
As mentioned earlier, the phase transition for identifiability/non-identifiability at $q = d$ coincides with the statistical physics uniqueness/non-uniqueness phase transition.
The lower bound when $q<d-\sqrt{d}+ \Theta(1)$ 
is also quite curious since it coincides exactly with the 
threshold for polynomial-time/NP-completeness for the decision problem~\cite{EHK,MR}.  In fact, the graph used in the proof of 
Part~\ref{part:col-exp} of Theorem~\ref{thm:col}
is inspired by the graph used in the NP-completeness proof in~\cite{EHK} and
the graph theoretic result in~\cite{MR}.

Since our results for the structure learning problem in the setting of $q$-colorings 
suggest an intimate connection 
between the efficiency/hardness of the learning problem and the uniqueness/non-uniqueness 
phase transition of the Gibbs distribution,
we further explore this connection for general $H$-colorings.
A sufficient condition for uniqueness on general graphs is the \textit{Dobrushin uniqueness 
	condition}~\cite{Dobrushin}, which is a standard tool in statistical physics
for establishing uniqueness of the Gibbs distribution.  

Dobrushin's condition considers the so-called influence matrix $R$ where the entry $R_{vw}$ 
measures the worst-case influence of vertex $w$ on $v$.  In particular, consider
all pairs of color assignments $\tau,\tau_w$ that differ only at vertex $w$.
$R_{vw}$ is the maximum (over these pairs $\tau,\tau_w$) of the difference
in total variation distance of the marginal distribution at $v$ conditional
on the color assignment $\tau(V \setminus v)$ versus $\tau_w(V \setminus v)$.  
Observe that $R_{vw}=0$ for non-adjacent pairs $v,w$.  
The Dobrushin uniqueness condition holds if the maximum row sum
in $R$ is strictly less than $1$, so that the total influence on a vertex of its neighborhood is less than $1$; see Definition~\ref{def:dob} for a precise definition.

We prove that if the Dobrushin uniqueness condition holds, then 
we can learn the underlying $n$-vertex graph in 
$\poly(n,q)$ time.

\begin{theorem}
	\label{thm:dob}
	Let $H \neq K_q^+$ be an arbitrary constraint graph.  	
	Suppose $G$ is such that 
	$\pi_G$ satisfies the Dobrushin uniqueness condition.
	Then, there is a structure learning algorithm
	that given $L \!=\! O (q^2 \log (\frac{n^2}{\varepsilon}))$ independent samples from $\pi_G$
	outputs 
	the graph $G$
	with probability at least $1-\varepsilon$
	and has running time~$O(L n^2)$.
\end{theorem}

\noindent
The above theorem provides an efficient structure learning algorithm under a fairly strong model assumption.
For soft-constraint systems, under a similar but weaker (i.e., easier to satisfy) condition, Bresler, Mossel and Sly \cite{BMS} 
give a structure learning algorithm with running time \textit{exponential} in the maximum degree $d$ of the graph $G$.
We provide a structure learning algorithm for hard-constraint systems with a similar running time that works under an even weaker assumption. 
Our algorithm works for
{\em permissive systems}, which 
are a class of hard-constraint models widely studied
in statistical physics; see, e.g.,~\cite{DSVW,MSW,DMS}. The precise definition, as well as the running time and sample complexity of our algorithm, are provided in Section \ref{sec:permissive}.

The rest of the paper is organized as follows. In Section \ref{sec:identify-unweight} we prove our characterization theorem (Theorem \ref{thm:identify}) and our learning lower bound for identifiable models (Theorem~\ref{thm:neg-identify}).
In Section \ref{sec:colorings}, we prove our results for colorings (Theorem~\ref{thm:col}). In particular, in Section~\ref{subsec:naive-alg} we introduce a general structure learning algorithm \textsc{structlearn-H} which will be the basis of all our algorithmic result.
Our $\poly(n,d,q)$-time algorithm under the Dobrushin uniqueness condition (Theorem \ref{thm:dob})
is established in Section \ref{sec:Dobrushin}.
In Section \ref{section:approx} we consider the approximation problem of learning a graph $G$ such that $\pi_G$ is close in total variation distance to the true distribution. 
Finally, the case of weighted $H$ and $G$ is considered in Section \ref{sec:weighted}.

\section{Identifiability}
\label{sec:identify-unweight}

As discussed in the introduction, 
given a constraint graph $H$, it is possible that $\pi_{G_1}=\pi_{G_2}$ for two distinct $H$-colorable graphs $G_1$ and $G_2$; i.e., the structure learning problem is not identifiable. 
In this section we  prove Theorem \ref{thm:identify} from the introduction that
characterizes the identifiable constraint graphs. 

Let $\mathcal{G}$ be a family of $H$-colorable graphs. Recall that a constraint graph $H$ is \emph{identifiable with respect to $\mathcal{G}$} if for any two distinct graphs $G_1,G_2\in\mathcal{G}$ we have $\pi_{G_1} \neq \pi_{G_2}$ or equivalently $\OO_{G_1} \neq \OO_{G_2}$. In particular, when $\mathcal{G}$ is the set of all finite $H$-colorable graphs we say that $H$ is identifiable; see Definition \ref{def:identify}.
Before proving Theorem \ref{thm:identify} we provide
the following useful alternative definition of identifiability.

\begin{lemma}
	\label{lemma:id_lem}
	A constraint graph $H \neq K_q^+$ is identifiable if and only if for any $H$-colorable graph $G=(V,E)$ and any two nonadjacent vertices $u,v \in V$ there is an $H$-coloring of $G$ where $u$ and $v$ are assigned incompatible colors.
\end{lemma}

\begin{proof}
	For the forward direction we prove the contrapositive.
	Let $G=(V,E)$ be an $H$-colorable graph and
	suppose there exists two nonadjacent vertices $u,v \in V$ 
	such that in every $H$-coloring $\sigma$ of $G$ these vertices
	receive compatible colors.
	Then, the graph $G$ and the graph $G' = (V,E \cup \{u,v\})$
	have the same set of $H$-colorings. Hence, $\pi_G = \pi_{G'}$ and so $H$ is not identifiable.  
	
	For the reverse direction,
	suppose that for every $H$-colorable graph $G=(V,E)$ and every pair of nonadjacent vertices $u,v \in V$ there exists an $H$-coloring of $G$ such that $u$ and $v$
	are assigned incompatible colors.
	Suppose also that for a pair of $H$-colorable graphs $G_1=(V,E_1)$ and $G_2=(V,E_2)$, we have
	$\pi_{G_1} = \pi_{G_2}$ (or equivalently that $\OO_{G_1}=\OO_{G_2}$). 
	We show that $G_1 = G_2$.
	First consider $\{u,v\} \not\in E_1$. Then, there exists an $H$-coloring $\tau \in \Omega_{G_1}$
	where $u$ and $v$ receive incompatible colors. Since also $\tau \in \Omega_{G_2}$, $\{u,v\} \not\in E_2$ . 
	Similarly, if $\{u,v\} \not\in E_2$, then $\{u,v\} \not\in E_1$. Thus, $G_1=G_2$ and so $H$ is identifiable.	 
\end{proof}
\noindent
To characterize identifiability of $H$-colorings
we previously introduced the supergraph $G_{ij}$ of $H$
obtained by duplicating the colors $i$ and $j$. 
The new copies of $i$ and $j$, denoted $i'$ and $j'$, have the same neighbors as $i$ and $j$, respectively, except for the one
edge $\{i',j'\}$ that is not present in $G_{ij}$---see Figure~\ref{fig:Hij} for examples of the graph $G_{ij}$.
We provide next the formal definition of $G_{ij}$.
\begin{definition}\label{dfn:Hij}
	Let $H=(V(H),E(H))$ be an arbitrary constraint graph with no self-loops.
	For each $\edge{i}{j}\in E(H)$, 
	we define the graph $G_{ij} = (V(G_{ij}),E(G_{ij}))$ as follows:
	\begin{enumerate}
		\item $V(G_{ij})=V(H)\cup\{i',j'\}$ where $i'$ and $j'$ are two new colors;
		\item If $\edge{a}{b} \in E(H)$, then the edge $\edge{a}{b}$ is also in $E(G_{ij})$;
		\item For each $k \in V(G_{ij}) \setminus \{i',j'\}$, 
		the edge $\edge{i'}{k}$ is in $G_{ij}$  
		if and only if the edge $\edge{i}{k}$ is in $H$, and 
		similarly $\edge{j'}{k}\in E(G_{ij})$ if and only if $\edge{j}{k} \in E(H)$;
	\end{enumerate}
\end{definition}

\begin{figure}[t]
	\begin{minipage}{0.22\linewidth}
		\centering
		\begin{tikzpicture}
		\node(3) at (0,0) {$3$};
		\node(1) at (-1.5,1.5) {$1$};
		\node(2) at (1.5,1.5) {$2$};
		\node(1') at (-1.5,2.6) {};
		\node(2') at (1.5,2.6) {};
		
		\draw (1)--(2);
		\draw (1)--(3);
		\draw (2)--(3);
		\end{tikzpicture}
		\caption*{$H$}
	\end{minipage}
	\begin{minipage}{0.22\linewidth}
		\centering
		\begin{tikzpicture}
		\node(3) at (0,0) {$3$};
		\node(1) at (-1.5,1.5) {$1$};
		\node(2) at (1.5,1.5) {$2$};
		\node(1') at (-1.5,2.5) {$1'$};
		\node(2') at (1.5,2.5) {$2'$};
		
		\draw (1)--(2);
		\draw (1')--(2);
		\draw (1)--(2');
		\draw (1)--(3);
		\draw (2)--(3);
		\draw (1')--(3);
		\draw (2')--(3);
		\end{tikzpicture}
		\caption*{$G_{12}$}
	\end{minipage}
	\qquad
	\begin{minipage}{0.22\linewidth}
		\centering
		\begin{tikzpicture}
		\node(3) at (-1.5,0) {$3$};
		\node(4) at (1.5,0) {$4$};
		\node(1) at (-1.5,1.5) {$1$};
		\node(2) at (1.5,1.5) {$2$};
		\node(1') at (-1.5,2.6) {};
		\node(2') at (1.5,2.6) {};
		
		\draw (1)--(2);
		\draw (1)--(3);
		\draw (2)--(4);
		\draw (3)--(4);
		\end{tikzpicture}
		\caption*{$H'$}
	\end{minipage}
	\begin{minipage}{0.22\linewidth}
		\centering
		\begin{tikzpicture}
		\node(3) at (-1.5,0) {$3$};
		\node(4) at (1.5,0) {$4$};
		\node(1) at (-1.5,1.5) {$1$};
		\node(2) at (1.5,1.5) {$2$};
		\node(1') at (-1.5,2.5) {$1'$};
		\node(2') at (1.5,2.5) {$2'$};
		
		\path
		(1) edge (2)
		(1') edge (2)
		(2') edge (1)
		(1) edge (3)
		(1') edge [bend right] (3)
		(2) edge (4)
		(2') edge [bend left] (4)
		(3) edge (4);
		\end{tikzpicture}
		\caption*{$G_{12}'$}
	\end{minipage}
	\caption{Two constraint graphs $H$ and $H'$ with corresponding supergraphs $G_{12}$ and $G_{12}'$.}
	\label{fig:Hij}
\end{figure}

Our characterization of identifiability, i.e., Theorem \ref{thm:identify}
from the introduction, is established by the next two lemmas. Lemma~\ref{lemma:id:no-loops} deals with the case of constraint graphs with at least one self-loop, while Lemma~\ref{lemma:id_thm} considers constraint graphs with no self-loops.

\begin{lemma}
	\label{lemma:id:no-loops}
	If $H \neq K_q^+$ has at least one self-loop, then $H$ is identifiable.
\end{lemma}

\begin{proof}
	The proof is divided into two cases corresponding to whether all vertices of $H$ have self-loops or not. 
	
	\medskip
	\noindent
	\textbf{Case 1:} \textit{At least one but not all vertices of $H$ have self-loops.} 
	Let $U$ be the set of vertices that have self-loops and let $W=V(H)\bl U$ be the set of vertices that do not. By assumption both $U$ and $W$ are not empty. Moreover, $U$ and $W$ are connected because by assumption $H$ is connected. 
	Thus, there exist $i\in U$ and $j\in W$ such that $\edge{i}{i},\edge{i}{j}\in E(H)$ and $\edge{j}{j} \not\in E(H)$.  We use this gadget to show that for any $H$-colorable graph $G=(V,E)$ and any two nonadjacent vertices $u,v \in V$ of $G$,
	there exists an $H$-coloring $\sigma$ of $G$
	where $u$ and $v$ are assigned incompatible colors. Then, by Lemma \ref{lemma:id_lem}, $H$ is identifiable.
	The $H$-coloring $\sS$ is defined as follows: $\sS(w)=i$ for all $w\neq u,v$ and $\sS(u)=\sS(v)=j$. Since $\edge{i}{i},\edge{i}{j} \in E(H)$,  $\sS$ is a valid $H$-coloring of $G$.
	Moreover, since $\edge{j}{j} \not\in E(H)$,  $u$ and $v$ receive incompatible colors and the result follows.
	
	\medskip
	\noindent
	\textbf{Case 2:} \textit{All vertices of $H$ have self-loops}. 
	Observe first that if $H$ is connected, $H \neq K_q^+$ and every vertex in $H$ has a self-loop,
	then there exist $i,j,k \in V(H)$ such that $\edge{i}{j}$, $\edge{j}{k}$, $\edge{i}{i}$, $\edge{j}{j}$, $\edge{k}{k}\in E(H)$ and $\edge{i}{k} \notin E(H)$. We use this gadget to show 
	for any $H$-colorable graph $G=(V,E)$ 
	and any pair of nonadjacent vertices $u,v \in V$
	there is an $H$-coloring $\sigma$ of $G$
	such that $(\sigma(u),\sigma(v)) \not\in E(H)$. Lemma \ref{lemma:id_lem} then implies that $H$ is identifiable.
	The $H$-coloring $\sS$ is given by: $\sS(w)=j$ for all $w\neq u,v$, $\sS(u)=i$ and $\sS(v)=k$. 
	Since color $j$ is compatible with colors $i$, $j$ and $k$ in $H$,
	$\sS$ is a valid $H$-coloring of $G$. Moreover,  $u$ and $v$ receive the incompatible colors $i$ and $k$ and so the result follows.
\end{proof}	
\vspace{-0.3em}

\begin{lemma}
	\label{lemma:id_thm}
	If $H \neq K_q^+$ has no self-loops, $H$ is identifiable if and only if for each $\{i,j\}\in E(H)$ there exists an $H$-coloring of $G_{ij}$ where $i'$ and $j'$ receive incompatible colors.
\end{lemma}

\begin{proof}
	Assume first that $H \neq K_q^+$ is identifiable and has no self-loops. 
	For every $\{i,j\}\in E(H)$, $G_{ij}$ is clearly $H$-colorable 
	(simply assign 
	color $k \in V(H)$ to the corresponding vertex in $G_{ij}$, color $i$ to $i'$ and color $j$ to $j'$).
	Hence, Lemma~\ref{lemma:id_lem} implies that there exists an $H$-coloring of $G_{ij}$ where $i'$ and $j'$ receive incompatible colors. This proves the forward direction of the lemma.
	
	For the reverse direction, suppose that for every $\{i,j\} \in E(H)$ there exists an $H$-coloring of $G_{ij}$ where $i'$ and $j'$ are assigned incompatible colors.
	Let $G=(V,E)$ be an arbitrary $H$-colorable graph. 
	We show that for every pair of nonadjacent vertices $u,v \in V$ in $G$
	there exists an $H$-coloring of $G$ where $u$ and $v$ receive incompatible colors.
	It then follows from Lemma \ref{lemma:id_lem} that $H$ is identifiable.
	
	Let $\sS$ be an $H$-coloring of $G$ 
	and let us assume that $\sS(u)$ and $\sS(v)$ are compatible colors. 
	(If $\sS(u)$ and $\sS(v)$ are incompatible colors in $H$, 
	there is nothing to prove.)
	We use $\sigma$ to construct an $H$-coloring $\sigma'$ where
	$u$ and $v$ receive incompatible colors.
	Let 
	$a = \sigma(u)$ and $b = \sigma(v)$.
	By assumption, there exists
	an $H$-coloring $\tau$ of $G_{ab}$ where the corresponding copies of $a$ and $b$, $a'$ and $b'$, receive incompatible colors. Define the $H$-coloring $\sS'$ of $G$ as follows:
	\begin{equation*}
	\sS'(w) = \tau(\sS(w)), \quad \forall w\neq u,v; \quad \sS'(u) = \tau(a');\quad \sS'(v) = \tau(b').
	\end{equation*}
	It is straightforward to check that $\sS'$ is a proper $H$-coloring of $G$. Since $u$ and $v$ receive incompatible colors in $\sigma'$ (i.e., $\tau(a')$ and $\tau(b')$), the proof is complete.
\end{proof}

\subsection{Learning lower bounds for identifiable models}

In this section we prove Theorem \ref{thm:neg-identify} from the introduction. In particular, we provide a constraint graph $F$ and a family of $F$-colorable $n$-vertex graphs of maximum degree $7$ such that the number of samples from $\pi_G^F$ required to learn any graph in this family, even with success probability $\exp(-O(n))$, is exponential in~$n$. 

We define the constraint graph $F=(V(F),E(F))$ first, which consists of an independent set of size $32$, denoted $I_{32}$, and four additional vertices $\{1,2,3,4\}$. Every vertex in the independent set $I_{32}$ is connected to these four vertices and also $\edge{1}{2},\edge{2}{3},\edge{3}{4} \in E(F)$; see Figure~\ref{fig:lb_id}(a).
\begin{figure}[t]
	\begin{minipage}{0.3\linewidth}
		\centering
		\begin{tikzpicture}[every loop/.style={}]
		\node(1) at (-2,0) {$1$};
		\node(2) at (-1,1.732) {$2$};
		\node(3) at (1,1.732) {$3$};
		\node(4) at (2,0) {$4$};
		\node[shape=circle, draw=black](5) at (0,0) {$I_{32}$};
		
		\path
		(1) edge [-] (2)
		(2) edge [-] (3)
		(3) edge [-] (4)
		(5) edge [ultra thick] (1)
		(5) edge [ultra thick] (2)
		(5) edge [ultra thick] (3)
		(5) edge [ultra thick] (4);
		\end{tikzpicture}
		\caption*{(a)}
	\end{minipage}
	\quad\quad\quad
	\begin{minipage}{0.55\linewidth}
		\centering
		\begin{tikzpicture}
		\node(a1)  [inner sep=1pt] at (0,0) {$a_1$};
		\node(a1') [inner sep=1pt] at (0,-1.5) {$a_1'$};
		\node(b1)  [inner sep=1pt] at (0,-3) {$b_1$};
		\node(c1)  [inner sep=1pt] at (0,-4.5) {$c_1$};
		
		\node(ai)  [inner sep=1pt] at (3,0) {$a_i$};
		\node(ai') [inner sep=1pt] at (3,-1.5) {$a_i'$};
		\node(bi)  [inner sep=1pt] at (3,-3) {$b_i$};
		\node(ci)  [inner sep=1pt] at (3,-4.5) {$c_i$};
		
		\node(ai1) [inner sep=1pt] at (5,-0) {$a_{i+1}$};
		\node(ai1')[inner sep=1pt] at (5,-1.5) {$a_{i+1}'$};
		\node(bi1) [inner sep=1pt] at (5,-3) {$b_{i+1}$};
		\node(ci1) [inner sep=1pt] at (5,-4.5) {$c_{i+1}$};
		
		\node(am)  [inner sep=1pt] at (8,0) {$a_m$};
		\node(am') [inner sep=1pt] at (8,-1.5) {$a_m'$};
		\node(bm)  [inner sep=1pt] at (8,-3) {$b_m$};
		\node(cm)  [inner sep=1pt] at (8,-4.5) {$c_m$};
		
		\node(x1) at (0,-2.25) {};
		\node(x2) at (3,-2.25) {};
		\node(x3) at (5,-2.25) {};
		\node(x4) at (8,-2.25) {};
		
		
		\path
		(a1) edge [bend left] (b1)
		(a1) edge [bend right] (c1)
		(a1') edge  (b1)
		(a1') edge [bend right] (c1)
		(b1) edge (c1)
		
		(ai1) edge (bi)
		(ai1) edge (ci)
		(ai1') edge (bi)
		(ai1') edge (ci)
		(bi1) edge (ai)
		(bi1) edge (ai')
		(bi1) edge (ci)
		(ci1) edge (ai)
		(ci1) edge (ai')
		(ci1) edge (bi)
		
		(x1) edge[white] node[black]{$\cdots\cdots$} (x2)
		(x3) edge[white] node[black]{$\cdots\cdots$} (x4);
		\end{tikzpicture}
		\caption*{(b)}
	\end{minipage}
	\caption{(a) The constraint graph $F$; each thick edge corresponds to $32$ edges, one incident to each vertex of $I_{32}$. (b) The graph $G_m$ showing the edges in $E_{i,i+1}$ and the connections between the vertices $a_1,a_1',b_1,c_1$.}
	\label{fig:lb_id}
\end{figure}
\begin{lemma}
	\label{lemma:ident-example}
	$F$ is identifiable.
\end{lemma}

\begin{proof}
	Let $G=(V,E)$ be an $F$-colorable graph. 
	Since $F$ is tripartite with a unique tripartition $\{\{1,3\}, \{2,4\}, I_{32}\}$, then so is $G$. Let $\{V_1,V_2,V_3\}$ be a tripartition of $G$
	and let $u,v$ be any two nonadjacent vertices of $G$. 
	We show that there is always an $F$-coloring of $G$
	where $u$ and $v$ receive incompatible colors. The result then follows from Lemma \ref{lemma:id_lem}.
	
	If $u$ and $v$ belong to the same $V_i$, 
	then by coloring all the vertices $V_1$ with color $1$, 
	all the vertices of $V_2$ with color $2$ 
	and all the vertices in $V_3$ with any color $c$ from $I_{32}$,
	we have a coloring of $G$ where $u$ and $v$ receive the same color.
	Since $F$ has no self-loops $u$ and $v$ are assigned incompatible colors.   
	
	If $u$ and $v$ belong to different $V_i$'s,
	suppose without loss of generality that $u \in V_1$ and $v \in V_2$.
	Consider the following $F$-coloring $\sS$ of $G$ where $c$ is any color from $I_{32}$: 
	\begin{equation*}
	\sS(w)=
	\left\{
	\begin{aligned}
	1\quad & \textrm{if}~w=u;\\
	4\quad &\textrm{if}~w=v;\\
	3\quad &\textrm{if}~w\in V_1\bl\{u\};\\
	2\quad &\textrm{if}~w\in V_2\bl\{v\};\\
	c\quad &\textrm{if}~w\in V_3.
	\end{aligned}
	\right.
	\end{equation*}
	In $\sS$, $u$ and $v$ receive the incompatible colors $1$ and $4$. Thus, we have shown that it is always possible to color nonadjacent vertices of $G$ with incompatible colors and the result follows from Lemma \ref{lemma:id_lem}.
\end{proof}
\noindent
We define next a family $\mathcal{G}_n$
of $n$-vertex graphs
of maximum degree $7$ such that every graph in the family have almost the same set of $F$-colorings. 
Every graph in $\mathcal{G}_n$ will be a supergraph of the graph $G_m =(V_m,E_m)$, whose vertex set is given by 
$$V_m = \{a_i,a'_i,b_i,c_i:1\leq i\leq m\}$$ 
with $m = n/4$.
(For clarity we assume first that $4$ divides $n$; we later explain how to adjust the definition of $\mathcal{G}_n$ when $n$ is not divisible by $4$.) 
For each $2 \le i \le m$, $\{a_i,a'_i,b_i,c_i\}$ is an independent set.
The edges with both endpoints in  $\{a_1,a'_1,b_1,c_1\}$ are: 
$$E_{1,1} := \{\{a_1,b_1\},\, \{a_1,c_1\},\,\{a'_1,b_1\},\, \{a'_1,c_1\},\{b_1,c_1\}\}.$$
The edges 
between the independent sets
$\{a_i,a'_i,b_i,c_i\}$ and~$\{a_{i+1},a'_{i+1},b_{i+1},c_{i+1}\}$ for $1\leq i < m$ are:
\begin{align*}
E_{i,i+1} :=
\big\{& {} \{a_i,b_{i+1}\},\{a_i,c_{i+1}\},\{a_i',b_{i+1}\},\{a_i',c_{i+1}\},\{b_i,a_{i+1}\},\\
&\{b_i,a'_{i+1}\},\{b_i,c_{i+1}\},
\{c_i,a_{i+1}\},\{c_i,a'_{i+1}\},\{c_i,b_{i+1}\}\}.
\end{align*}
We then let
$ E_m = \left( \bigcup_{i=1}^{m-1} E_{i,i+1} \right) \cup E_{1,1}$; see Figure \ref{fig:lb_id}(b). 

Now, let 
\begin{equation}\label{eq:M}
M = \{\{a_i,b_{i+2}\}:i=1,2,\ldots,m-2\}
\end{equation}
and let 
$E^{(1)},E^{(2)},\dots,E^{(t)}$ be all the possible subsets of $M$; hence $t=2^{m-2}$.
We define $\mathcal{G}_n$ as:
$$\mathcal{G}_n = \{G^{(1)}=(V_m,E_m\cup E^{(1)}),G^{(2)}=(V_m,E_m \cup E^{(2)}),\dots,G^{(t)}=(V_m,E_m\cup E^{(t)})\}.$$
Since the maximum degree of $G_m$ is $6$, every graph in $\mathcal{G}_n$ has maximum degree 7. Moreover, if we let $A_{m}=\{a_i,a'_i:1\leq i\leq m\}$, $B_{m}=\{b_i:1\leq i\leq m\}$ and $C_{m}=\{c_i:1\leq i\leq m\}$, then it is clear from our construction that $(A_{m},B_{m},C_{m})$ is a tripartition for every graph in $\mathcal{G}_n$. An immediate consequence of this is that every graph in $\mathcal{G}_n$ has an $F$-coloring that assigns, for example, color $2$ to every vertex in $A_m$, color $3$ to every vertex in $B_m$ and a color from the independent set $I_{32}$ to every vertex in $C_m$. Therefore, all graphs in the family $\mathcal{G}_n$ are $F$-colorable.


The next theorem shows that structure learning is hard for $F$ and $\mathcal{G}_n$.
\begin{theorem}
	\label{thm:ident-lower-bound}
	Let $m \in \N^+$ such that $m \ge 2$ and let $n=4m$.
	Then, any structure learning algorithm 
	for the constraint graph $F$ and the family of graphs $\mathcal{G}_n$ 
	that succeeds with probability at least $2^{-(m-3)}$ requires at least 
	$2^{m+1}$ samples.
\end{theorem}
\noindent
Before proving this theorem we state two key facts that will be used in its proof.
In particular, Fact~\ref{fact:tripartition} shows that actually $(A_{m},B_{m},C_{m})$ is the unique tripartition of $G_m$, and
Fact \ref{fact:lb_eta} gives a lower bound for the number of samples required 
by a structure learning algorithm
to guarantee a prescribed success probability.

\begin{fact}
	\label{fact:tripartition}
	For any $m \in \N^+$ the graph $G_m$ is tripartite. Moreover, it has a unique 
	tripartition $(A_{m},B_{m},C_{m})$, where  $A_{m}=\{a_i,a'_i:1\leq i\leq m\}$, $B_{m}=\{b_i:1\leq i\leq m\}$ and $C_{m}=\{c_i:1\leq i\leq m\}$. 
\end{fact}

\begin{fact}
	\label{fact:lb_eta}
	Let $H$ be an arbitrary constraint graph.
	Suppose $\mathcal{\hat{G}}=\{\hat{G}_1,\hat{G}_2,\ldots,\hat{G}_r\}$ is a family of $r$ distinct $H$-colorable graphs such that $H$ is identifiable with respect to $\mathcal{\hat{G}}$.
	Assume also that $\hat{G}_1$ is a subgraph of $\hat{G}_i$ for all $2\leq i\leq r$ and that $\hat{G}_r$ is a supergraph of $\hat{G}_i$ for all $1\leq i\leq r-1$. Let
	\begin{equation*}
	\eta = 1-\frac{|\OO_{\hat{G}_r}|}{|\OO_{\hat{G}_1}|}.
	\end{equation*}
	If there exists a structure learning algorithm for $H$ and $\mathcal{\hat{G}}$ such that for any $G^*\in \mathcal{\hat{G}}$, given $L$ independent samples from $\pi_{G^*}^H$ as input, it outputs $G^*$ with probability at least $1/r+\aA$ with $\alpha > 0$, then $L\geq \aA/\eta$. 
\end{fact}
\noindent
We observe that when $\alpha = 0$, the structure learning algorithm 
that outputs a graph from $\mathcal{\hat{G}}$ uniformly at random
has success probability $1/r$
without requiring any samples.
We are now ready to prove Theorem~\ref{thm:ident-lower-bound}.

\newenvironment{myproof-thm:ident-lower-bound}{\paragraph{\textbf{Proof of Theorem \ref{thm:ident-lower-bound}}}}{}
\begin{myproof-thm:ident-lower-bound}
	Let $m \ge 2$, $n=4m$ and for ease of notation let  $\mathcal{G}=\mathcal{G}_n$. Let $G^{(1)}=G_m$ and $G^{(t)}=G_{m}\cup M$ where $M$ is defined in (\ref{eq:M}) and $t=2^{m-2}$. Hence,
	the graph
	$G^{(1)}$ (resp., $G^{(t)}$) is a subgraph (resp., supergraph) of every other graph in $\mathcal{G}$.
	Moreover, all graphs in $\mathcal{G}$ are distinct
	and $F$ is identifiable by Lemma~\ref{lemma:ident-example}. 
	Hence, to apply Fact~\ref{fact:lb_eta} all we need is a lower bound for $\eta =1-|\OO_{G^{(t)}}|/|\OO_{G^{(1)}}|$.
	
	By Fact \ref{fact:tripartition}, $G_m$ has a unique tripartition $(A_m,B_m,C_m)$. 
	Since $\{\{1,3\},\{2,4\},I_{32}\}$ is a tripartition for $F$, every $F$-coloring of $G_m$ induces the same tripartition of $G_m$. 
	That is to say, in every $F$-coloring of $G_m$
	one of the sets $A_m$, $B_m$ or $C_m$ is colored with colors $\{1,3\}$, another is colored with $\{2,4\}$ and the third one is colored using colors from the independent set $I_{32}$ of $F$.
	Then, the number of $F$-colorings of $G_m$ such that $A_{m}$ is colored with colors from $I_{32}$ is $K \cdot 32^{2m}$, 
	where $K$ is the number of $F$-colorings of $B_m \cup C_m$ given  a fixed $F$-coloring of $A_m$ that only uses colors from $I_{32}$.
	Observe that $K$ is the same for every $F$-coloring of $A_m$, and $K\geq 2$ since we can always color $B_m$ with color $2$ and $C_m$ with $3$, or color $B_m$ with $3$ and $C_m$ with $2$.
	On the other hand, the number of $F$-colorings where $B_{m}$ 
	receives colors from $I_{32}$ is at most $2\cdot 32^m \cdot 2^m \cdot 2^{2m} = 2\cdot 32^m \cdot 8^m$, and similarly for $C_{m}$. Hence, the probability that in a uniformly random $F$-coloring of $G_m$, $A_{m}$ is colored with colors from the independent set $I_{32}$ is at least
	\[\frac{K\cdot 32^{2m}}{K\cdot 32^{2m}+4\cdot 32^m \cdot 8^m} = 1 - \frac{4 \cdot 8^{m}}{K\cdot 32^{m}+4\cdot 8^m } \geq 1 - \frac{1}{2^{2m-1}}.
	\]
	
	Let $\sS$ be an $F$-coloring of $G_m$ such that $A_{m}$ is colored with colors from $I_{32}$. Since colors from $I_{32}$ are compatible with any other color, the pair of colors $\sS(a_i),\sS(b_{i+2})$ are compatible for any $1\leq i\leq m-2$. Therefore, $\sS$ is also a valid $F$-coloring of $G^{(t)}$ and thus $\sS\in\OO_{G^{(t)}}$. We then deduce that
	\begin{equation*}
	1-\eta = \frac{|\OO_{G^{(t)}}|}{|\OO_{G^{(1)}}|}  \geq 1-\frac{1}{2^{2m-1}}.
	\end{equation*}
	Since $|\mathcal{G}|=2^{m-2}$,
	it follows from Fact~\ref{fact:lb_eta} that the number of samples required to learn a graph in $\mathcal{G}$ with success probability $2^{-(m-3)}$ is at least
	\begin{equation*}
	L \geq \frac{2^{-(m-3)}-2^{-(m-2)}}{\eta} \geq 2^{m+1}.\tag*{$\square$}
	\end{equation*}
\end{myproof-thm:ident-lower-bound}

\noindent
To finish the proof of Theorem \ref{thm:ident-lower-bound} we provide the proofs of Facts~\ref{fact:tripartition} and \ref{fact:lb_eta}.

\newenvironment{myproof-fact:tripartition}{\paragraph{\textbf{Proof of Fact \ref{fact:tripartition}}}}{\hfill$\square$}
\begin{myproof-fact:tripartition}
	We prove this by induction. 
	$G_1$ has exactly one tripartition $(\{a_1,a_1'\},b_1,c_1)$. 
	Suppose inductively that
	$(A_{m-1},B_{m-1},C_{m-1})$
	is the only tripartition of $G_{m-1}$. 
	Since $\{a_i,b_{i-1}\}$, $\{a_i,c_{i-1}\}\in E_m$, $a_m$ belongs to $A_{m}$ in any tripartition of $G_m$. Similar statements hold for $a'_m$, $b_m$ and $c_m$ as well. Therefore, $(A_{m},B_{m},C_{m})$ is the unique tripartition of $G_m$.
\end{myproof-fact:tripartition}

\newenvironment{myproof-fact:lb_eta}{\paragraph{\textbf{Proof of Fact \ref{fact:lb_eta}}}}{\hfill$\square$}
\begin{myproof-fact:lb_eta}
	Let $\mathcal{A}$ be any (possibly randomized) structure learning algorithm that given $L$ independent samples $\Gamma = (\sS^{(1)},\ldots,\sS^{(L)}) \in \OO_{G^*}^L$ from an unknown distribution $\pi_{G^*}$ for some $G^*\in\mathcal{\hat{G}}$, outputs a graph $\mathcal{A}(\Gamma)$ in $\mathcal{\hat{G}}$. 
	For any $G^*\in\hat{\mathcal{G}}$, the probability that $\mathcal{A}$ learns the graph correctly given $L$ independent samples from $\pi_{G^*}$ is
	\begin{align*}
	\Pr [ \mathcal{A}(\Gamma) = G^* ]
	={}\sum_{x \in \OO_{G^*}^L} {\Pr}_{\pi_{G^*}} [\Gamma=x]
	\Pr[\mathcal{A}(x) = G^*].
	\end{align*}
	Since $\hat{G}_r$ is a supergraph of $G^*$, we have $\OO_{\hat{G}_r}\subseteq \OO_{G^*}$. Let $\mathcal{T}$ be the set of all sample sequences $\sS^{(1)},\ldots,\sS^{(L)}$
	such that $\sS^{(i)} \not\in \OO_{\hat{G}_r}$ for at least one $i$; namely, $\mathcal{T} = \OO_{G^*}^L \bl \OO_{\hat{G}_r}^L$. 
	Note that $|\mathcal{T}| = |\OO_{G^*}|^L - |\OO_{\hat{G}_r}|^L$.
	Then,
	\begin{align*}
	{\Pr} [ A(\Gamma) = G^* ]
	=& \sum_{x \in \OO_{\hat{G}_r}^L} \frac{1}{|\OO_{G^*}|^L} \cdot \Pr[\mathcal{A}(x) = G^*] + \sum_{x \in \mathcal{T}} \frac{1}{|\OO_{G^*}|^L} \cdot \Pr[\mathcal{A}(x) = G^*]\\
	\leq& \frac{1}{|\OO_{\hat{G}_r}|^L} \sum_{x \in \OO_{\hat{G}_r}^L} \Pr[\mathcal{A}(x) = G^*] +  \frac{|\OO_{G^*}|^L - |\OO_{\hat{G}_r}|^L}{|\OO_{G^*}|^L}.
	\end{align*}
	Since $\hat{G}_1$ is a subgraph of $G^*$, we have $\OO_{G^*} \subseteq \OO_{\hat{G}_1}$. Thus,
	\begin{equation*}
	\frac{|\OO_{G^*}|^L - |\OO_{\hat{G}_r}|^L}{|\OO_{G^*}|^L} = 1 - \frac{|\OO_{\hat{G}_r}|^L}{|\OO_{G^*}|^L} \leq 1 - \frac{|\OO_{\hat{G}_r}|^L}{|\OO_{\hat{G}_1}|^L} = 1-(1-\eta)^L \leq L\eta. 
	\end{equation*}
	Suppose the structure learning algorithm $\mathcal{A}$ has success probability at least $1/r+\alpha$ for any $G^*\in\hat{\mathcal{G}}$; that is,
	\begin{equation*}
	{\Pr} [ \mathcal{A}(\Gamma) = G^* ] \geq \frac{1}{r}+\aA, \quad \forall G^*\in\hat{\mathcal{G}}.
	\end{equation*}
	Then,
	\begin{equation*}
	\frac1r+\aA \leq \frac{1}{|\OO_{\hat{G}_r}|^L} \sum_{x \in \OO_{\hat{G}_r}^L} \Pr[\mathcal{A}(x) = G^*] + L\eta, \quad \forall G^*\in\hat{\mathcal{G}}.
	\end{equation*}
	Since $\sum_{G^*\in\hat{\mathcal{G}}} \pr[\mathcal{A}(x) = G^*] = 1$ for any fixed sample sequence $x$, summing up over $\hat{\mathcal{G}}$ we get
	\begin{gather*}
	1+r\aA \leq \frac{1}{|\OO_{\hat{G}_r}|^L} \sum_{x \in \OO_{\hat{G}_r}^L} \sum_{G^*\in\hat{\mathcal{G}}} \Pr[\mathcal{A}(x) = G^*] + rL\eta = 1+rL\eta.
	\end{gather*}
	Hence, $L\geq\aA/\eta$ as claimed.
\end{myproof-fact:lb_eta}

\subsection{\texorpdfstring{Proof of Theorem \ref{thm:neg-identify}}{Proof of Theorem 2}}
To conclude this section, we provide the proof of Theorem~\ref{thm:neg-identify} from the introduction, which follows straightforwardly from Theorem~\ref{thm:ident-lower-bound}.
\newenvironment{myproof-thm:neg-identify}{\paragraph{\textbf{Proof of Theorem \ref{thm:neg-identify}}}}{\hfill$\square$}
\begin{myproof-thm:neg-identify}
	Let $n \ge 8$. If $4$ divides $n$,
	then Theorem \ref{thm:ident-lower-bound}
	implies there exists a constant $c>0$ such that
	any structure learning algorithm for the constraint graph $F$ and the graph family $\mathcal{G}_n$ 
	with success probability at least $\exp(-cn))$ 
	requires at least $\exp(cn)$ samples. Since  $\mathcal{G}_n \subseteq \mathcal{G}(n,7)$, the result follows.
	
	The same ideas carry over without significant modification to the case when $4$ does not divide $n$.
	For example, 
	suppose that $n = 4m + 1$ for some $m \ge 2$.
	For $G \in \mathcal{G}_{4(m+1)}$, let $\hat{G}$ be the subgraph of $G$ induced by $V_{m+1} \setminus \{a_{m+1}',b_{m+1},c_{m+1}\}$.
	We define $\mathcal{G}_n$ as the family of the subgraphs $\hat{G}$ for each $G \in \mathcal{G}_{4(m+1)}$.
	When $n=4m+2$ or $n=4m+3$, we consider instead the graph families of the subgraphs induced by  $V_{m+1} \setminus \{b_{m+1},c_{m+1}\}$ and $V_{m+1} \setminus \{c_{m+1}\}$, respectively.    
	The argument in the proof of Theorem \ref{thm:ident-lower-bound} carries over to these graph families straightforwardly.
	Since in every case $\mathcal{G}_n \subseteq \mathcal{G}(n,7)$ the result follows.	  
\end{myproof-thm:neg-identify}

\section{Learning proper $q$-colorings}
\label{sec:colorings}

In this section we consider statistical identifiability, structure learning and  equivalent-structure learning for proper $q$-colorings, where $H = K_q$ and $\pi_G$ is the uniform distribution over the proper $q$-colorings of the graph $G$. In particular, we prove Theorem~\ref{thm:col} from the introduction.

\subsection{A structure learning algorithm}
\label{subsec:naive-alg}
In this subsection we introduce a general structure learning algorithm for any constraint graph $H$ with at least one hard constraint; i.e., $H\neq K_q^+$. 
In Section~\ref{subsection:efficient-coloring}, we analyze its running time and sample complexity for proper colorings. Later in Sections~\ref{sec:Dobrushin} and \ref{sec:permissive}, we consider more general settings where this algorithm is also efficient.

Fix $H \neq K_q^+$ and
suppose  $\edge{i}{j} \not\in E(H)$. Given independent samples $\sigma^{(1)},\dots,\sigma^{(L)}$ from $\pi_G=\pi_G^H$ for some unknown graph $G=(V,E)$,
the algorithm checks for every pair of vertices $u,v \in V$
whether there is at least one sample $\sigma^{(k)}$ such that
$\sigma^{(k)}_u = i$ and $\sigma^{(k)}_v = j$.
If this is the case, then the edge $\edge{u}{v}$ does not belong to $E$.
Otherwise, the algorithm adds the edge $\edge{u}{v}$ to the estimator $\hat{E}$ of $E$.
This structure learning algorithm, which we call \textsc{structlearn-H}, has running time $O(Ln^2)$ and was used before in \cite{BGS-hc} for the hard-core model. The effectiveness of \textsc{structlearn-H} depends crucially on how likely are nonadjacent vertices to receive colors $i$ and $j$.
For $v \in V$, let
$X_v$ be the random variable for the color of $v$ under $\pi_G$.

\begin{lemma}
	\label{lemma:hard-contraints:main}
	Let $H\neq K_q^+$ and $\edge{i}{j} \not\in E(H)$. Suppose that for all $\edge{u}{v} \not\in E$
	$$\Pr[X_u=i,X_v=j] \ge \delta,$$
	for some $\delta > 0$. 
	Let $\hat{G} = (V,\hat{E})$ be the output of the algorithm \textsc{structlearn-H}.
	Then, for all $\varepsilon \in (0,1)$,
	$
	\Pr[E=\hat{E}] \ge 1-\varepsilon
	$
	provided $L \ge 8\delta^{-1} \log (\frac{n^2}{2\varepsilon})$.
\end{lemma}
\begin{proof}
	Suppose $\edge{u}{v} \not\in E$ and let $Z_{uv}$ be the number of samples 
	where vertices $u$ and $v$ are assigned colors $i$ and $j$, respectively. Since $\E[Z_{uv}] \ge \delta L $, a Chernoff bound implies
	\[\Pr[Z_{uv} = 0] \le \Pr\left[Z_{uv} \le \frac{\delta L}{2}\right]  \le \exp\left(\frac{-\delta L}{8} \right) \le \frac {2\varepsilon}{n^2}.\]
	The result follows from a union bound over the edges.
\end{proof}

\subsection{\texorpdfstring{Efficient structure learning when $q \geq d+1$}{Efficient structure learning when q >= d+1}}
\label{subsection:efficient-coloring}

In this subsection we prove Part~\ref{part:1} of Theorem~\ref{thm:col}. 
We show that for proper $q$-colorings with $q\geq d+1$ and any graph in $\mathcal{G}(n,d)$, 
the structure learning algorithm \textsc{structlearn-H} (see Section \ref{subsec:naive-alg}) 
requires $O(qd^3\log{(n/\varepsilon)})$ samples to 
succeed with probability at least $1-\varepsilon$ and
has running time $O(qd^3n^2\log{(n/\varepsilon)})$.
This can be deduced immediately from the next lemma and Lemma~\ref{lemma:hard-contraints:main}, since $H=K_q\neq K_q^+$ in this setting.

\begin{lemma}
	\label{lemma:colorings:main}
	Suppose that $q\geq d+1$ and
	let $\edge{u}{v} \not\in E$.
	Then 
	$$
	\Pr[X_u=X_v] \geq\frac{1}{q( d+1)^3}.
	$$
\end{lemma}

\begin{proof}
	Let $u,v \in V$ be such that $\edge{u}{v} \not\in E$, 
	and let $\partial u$, $\partial v$ denote the neighborhoods of $u$ and $v$, respectively, which may
	overlap. Let $G'$ be $G$ with the vertices $u,v$ removed (the edges adjacent to $u$ and $v$ are removed as well). 
	We will generate a uniformly random coloring of $G$ using rejection sampling as follows. Pick a uniformly random 
	coloring of $G'$, a uniformly random color $c_1$ for $u$, and a uniformly random color $c_2$ for $v$. 
	If the resulting coloring is valid for $G$ then accept, otherwise reject. 
	Since in every round each coloring has the same probability of being picked, the generated coloring is a uniformly random coloring of $G$. 
	
	Let $A(c,s,t)$ be the set of colorings of $G'$ where color $c$ appears exactly $s$ times in $\partial u$ and exactly~$t$ times in $\partial v$.
	Given a coloring in $A(c,s,t)$ we produce a coloring in $A(c,0,0)$ as follows. List the positions where $c$ occurs in $\partial u\cup \partial v$
	and then re-color the vertices in the order of the list. Note that every vertex $w$ in $\partial u\cup \partial v$ has at least $2$ colors that 
	do not occur in its neighborhood since in $G'$ the vertex $w$ has degree at most $ d-1$ (recall that we removed $u$ and $v$ from $G$ to obtain $G'$).
	This maps at most $ d^{s+t}$ colorings from $A(c,s,t)$ to a coloring in $A(c,0,0)$ (given the list of positions we can recover
	the original coloring; there are at most $ d^{s+t}$ lists where we first list the vertices in $\partial u$ and then vertices in 
	$\partial v\setminus \partial u$). Hence we have
	\begin{equation}\label{ine1}
	| A(c,0,0) | \geq \frac{| A(c,s,t) | }{  d^{s+t}}.
	\end{equation}
	
	Let $A(c,\leq j,\leq k) := \sum_{s\leq j,t\leq k} A(c,s,t)$ and let $\mu$ be the uniform distribution over the colorings of $G'$. We claim that in any coloring there exists at least $q-d$ colors that satisfy the following: 
	$c$ occurs at most once in $\partial u$ and at most once in $\partial v$.
	Indeed, adding over all colors the number 
	of occurrences in  $\partial u$ and the number of occurrences in $\partial v$ we can get at most $2 d$;
	thus at most $ d$ colors can occur at least twice in $\partial u$ or at least twice in $\partial v$.
	Thus 
	$$
	\mu(A(1,\leq 1,\leq 1)) + \mu(A(2,\leq 1,\leq 1)) + \cdots + \mu(A(q,\leq 1,\leq 1)) \geq q- d, 
	$$
	and by symmetry 
	$\mu(A(1,\leq 1,\leq 1)) \geq \frac{q- d}{q}\geq \frac{1}{ d+1}$.
	Since
	$$\mu(A(1,\leq 1,\leq 1)) = \sum_{i,j \in \{0,1\}} \mu(A(1,i,j)),$$
	from~\eqref{ine1} we get 
	$$
	\mu(A(1,0,0)) \geq \frac{1}{(d+1)^3}.
	$$
	Therefore with probability at least $1/(d+1)^3$ color $1$ does not occur in $\partial u\cup \partial v$ 
	in a random coloring of $G'$. With probability $1/q^2$ we propose $c_1=c_2=1$ in the rejection sampling procedure and hence with probability 
	at least 
	$
	\frac{1}{q^2( d+1)^3}
	$ 
	our process accepts and produces a coloring of $G$ where $u,v$ both receive color~$1$. Thus in a uniformly 
	random coloring of $G$ vertices $u,v$ receive color $1$ with probability at least 
	$\frac{1}{q^2( d+1)^3}$, and by symmetry the probability that they receive the same color is at least
	$\frac{1}{q( d+1)^3}$, as claimed.
\end{proof}



\subsection{Identifiability for proper $q$-colorings}
\label{subsection:identifiability}
In this subsection we prove Part~\ref{part:2} of Theorem~\ref{thm:col}. We show that when $q \le d$ there exist two distinct graphs $G,G' \in \mathcal{G}(n,d)$ such that $\pi_G = \pi_{G'}$, or equivalently that $G$ and $G'$ have the same set of $q$-colorings.

\begin{theorem}\label{thm:id-col}
	Let $n,q,d\in\mathbb{N}^+$ such that $q\leq d$ and $n \ge q+2$. Then, the structure learning problem for $q$-colorings is not identifiable with respect to the family of graphs $\mathcal{G}(n,d)$.
\end{theorem}

\begin{proof}
	Let $G=(V,E)$ be a graph with 
	$$V=\{c_1,\ldots,c_{q-1},u,v,w_0,\ldots,w_{n-q-2}\},$$
	where  
	$\{c_1,\ldots,c_{q-1},u,v\}$ is a clique of size $q+1$ except for the one edge $\edge{u}{v}$
	that is not in $E$, and 
	$\{w_0,\ldots,w_{n-q+2}\}$ is a simple path from $w_0$ to $w_{n-q-2}$. 
	$G$ has one additional edge connecting $v$ and $w_0$.
	Then, in every $q$-coloring of $G$ the vertices $u$ and $v$ receive the same color,
	and so $u$ and $w_0$ are assigned distinct colors.
	Hence, the graph $G$ and the graph $G' =(V,E \cup \edge{u}{w_0})$ have the same set of $q$-colorings.
	Since both $G$ and $G'$ are $n$-vertex graphs of maximum degree at most $q \le d$, the structure learning for $q$-colorings is not identifiable with respect to $\mathcal{G}(n,d)$.
\end{proof}

\subsection{\texorpdfstring{Strong lower bound when $q < d -\sqrt{d} + \Theta(1)$}{Strong lower bound}}
\label{subsection:strong-lower}
In this subsection we prove Part~\ref{part:col-exp} of Theorem~\ref{thm:col}, establishing a strong learning lower bound for proper colorings when $q < d -\sqrt{d} + \Theta(1)$. 

As previously defined, an equivalent-structure learning algorithm for a graph family $\mathcal{G}$ finds a graph $\hat{G} \in \mathcal{G}$ such that $\OO_{G}=\OO_{\hat{G}}$, where $G \in \mathcal{G}$ is the actual hidden graph. 
We exhibit a family of graphs of maximum degree $q+\sqrt{q}+\TT(1)$ 
such that every graph in the family has almost the same set of $q$-colorings. This makes equivalent-structure learning hard in this family. We use this fact to prove Part \ref{part:col-exp} of Theorem~\ref{thm:col}.

We define first a graph $G_{m,t}=(V_{m,t},E_{m,t})$ and every graph in our graph family will be a supergraph of $G_{m,t}$.
For any $m,t\in \mathbb{N}^+$ with $t<q$, 
the graph $G_{m,t}$ is defined as follows.
Let $C_1,\ldots,C_m$ and $C_1',\ldots,C_m'$ be cliques of size $q-1$, and let $I_1,\ldots,I_m$ and $I_1',\ldots,I_m'$ be independent sets of size $t$. 
Moreover, let $s_1,\ldots,s_m$ be $m$ additional vertices.
Then,
$$V_{m,t} = \bigcup_{i=1}^{m} \{V(C_i),V(C_i'),V(I_i),V(I_i'),s_i\},$$
where $V(C_i), V(C_i'), V(I_i), V(I_i')$ are the vertices of $C_i, C_i', I_i, I_i'$, respectively.
In addition to the edges in the cliques $C_i$ and $C_i'$ for $1 \le i \le m$,
$E_{m,t}$ contains the following edges:
\begin{enumerate}
	\item For $1\leq i\leq m$, there is a complete bipartite graph between $C_i$ and $I_i$. That is, for $u\in C_i$ and $v\in I_i$, $\edge{u}{v} \in E_{m,t}$. 
	\item 
	For $2\leq i\leq m$, each $C_i$ is partitioned into $t$ almost-equally-sized disjoint subsets $C_{i,1},\ldots,C_{i,t}$ of size either $\lfloor (q-1)/t \rfloor$ or $\lceil (q-1)/t \rceil$. 
	Then, the $j$-th vertex of $I_{i-1}$ is connected to every vertex in $C_{i,j}$. 
	\item Edges between $C_1',\ldots,C_m'$ and $I_1',\ldots,I_m'$ are defined in exactly the same manner.
	\item For $1\leq i\leq m$, the vertex $s_i$ is adjacent to exactly one vertex in $I_i$ and to exactly one in $I_i'$; 	
\end{enumerate}
see Figure \ref{fig:lb_coloring} for an illustration of the graph $G_{m,t}$.
The key fact about the graph $G_{m,t}$ that allows us to construct a graph family with the desired properties is the following.
\begin{lemma}\label{lemma:lb_coloring}
	Let $q,m,t\in \mathbb{N}^+$ and $q\geq 3$. In every $q$-coloring of $G_{m,t}$ all cliques $C_1,\ldots,C_m$ are colored by the same set of $q-1$ colors, and all independent sets $I_1,\ldots,I_m$ are colored with the remaining color. The same holds for $C_1',\ldots,C_m'$ and $I_1',\ldots,I_m'$ .
\end{lemma}
\noindent
Lemma~\ref{lemma:lb_coloring} implies that every $q$-coloring of $G_{m,t}$ is determined by the colors of $I_1,I_1'$ and those of the vertices $s_1,\ldots,s_m$. 

\begin{figure}[tb]
	\centering
	\begin{tikzpicture}
	\node(C1) [shape=circle, draw=black, minimum size=1.5cm] at (0,0) {$C_1$};
	\node(I1) [shape=circle, draw=black, minimum size=0.9cm] at (2,0) {$I_1$};
	\node(C2) [shape=circle, draw=black, minimum size=1.5cm] at (4,0) {$C_2$};
	\node(I2) [shape=circle, draw=black, minimum size=0.9cm] at (6,0) {$I_2$};
	\node(Im1)[shape=circle, draw=black, minimum size=0.9cm] at (9,0) {$I_{m-1}$};
	\node(Cm) [shape=circle, draw=black, minimum size=1.5cm] at (11,0) {$C_m$};
	\node(Im) [shape=circle, draw=black, minimum size=0.9cm] at (13,0) {$I_m$};
	
	\node(C1') [ shape=circle, draw=black, minimum size=1.5cm] at (0,-3) {$C_1'$};
	\node(I1') [shape=circle, draw=black, minimum size=0.9cm] at (2,-3) {$I_1'$};
	\node(C2') [ shape=circle, draw=black, minimum size=1.5cm] at (4,-3) {$C_2'$};
	\node(I2') [shape=circle, draw=black, minimum size=0.9cm] at (6,-3) {$I_2'$};
	\node(Im1')[shape=circle, draw=black, minimum size=0.9cm] at (9,-3) {$I_{m-1}'$};
	\node(Cm') [shape=circle, draw=black, minimum size=1.5cm] at (11,-3) {$C_m'$};
	\node(Im') [shape=circle, draw=black, minimum size=0.9cm] at (13,-3) {$I_m'$};
	
	\node(s1) at (2,-1.5) {$s_1$};
	\node(s2) at (6,-1.5) {$s_2$};
	\node(sm1) at (9,-1.5) {$s_{m-1}$};
	\node(sm) at (13,-1.5) {$s_m$};
	\node(i1) [fill, circle, inner sep=1pt] at (2,-0.3) {};
	\node(i1')[fill, circle, inner sep=1pt] at (2,-2.7) {};
	\node(i2) [fill, circle, inner sep=1pt] at (6,-0.3) {};
	\node(i2')[fill, circle, inner sep=1pt] at (6,-2.7) {};
	\node(im1) [fill, circle, inner sep=1pt] at (9,-0.3) {};
	\node(im1')[fill, circle, inner sep=1pt] at (9,-2.7) {};
	\node(im) [fill, circle, inner sep=1pt] at (13,-0.3) {};
	\node(im')[fill, circle, inner sep=1pt] at (13,-2.7) {};
	
	\path
	(C1) edge (I1)
	(I1) edge [dashed] (C2)
	(C2) edge (I2)
	(Im1) edge [dashed] (Cm)
	(Cm) edge (Im)
	
	(C1') edge (I1')
	(I1') edge [dashed] (C2')
	(C2') edge (I2')
	(Im1') edge [dashed] (Cm')
	(Cm') edge (Im')
	
	(s2) edge[white] node[black]{$\cdots\cdots$} (sm1)
	
	(s1) edge (i1)
	(s1) edge (i1')
	(s2) edge (i2)
	(s2) edge (i2')
	(sm1) edge (im1)
	(sm1) edge (im1')
	(sm) edge (im)
	(sm) edge (im');
	\end{tikzpicture}
	\caption{The graph $G_{m,t}$. Each of $C_1,\ldots,C_m,C_1',\ldots,C_m'$ is a clique of size $q-1$ and each of $I_1,\ldots,I_m,I_1',\ldots,I_m'$ is an independent set of size $t<q$. Solid lines between two clusters mean that every vertex from one cluster is adjacent to every vertex in the other cluster. Dashed lines between $I_{i-1}$ and $C_i$ mean that every vertex in $I_{i-1}$ is adjacent to roughly $(q-1)/t$ vertices in $C_i$ with no two vertices in $I_{i-1}$ sharing a common neighbor in $C_i$.}
	\label{fig:lb_coloring}
\end{figure}
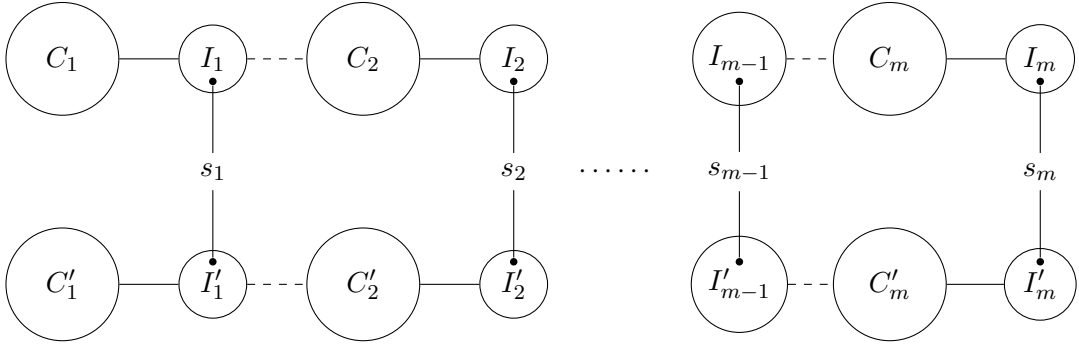

We define next the family of graphs $\mathcal{G}_{m,t}$.
All the graphs in this family are distinct and are supergraphs of $G_{m,t}$. For each $1 \le i \le m$ choose (and fix) a pair of vertices $x_i \in C_i$ and $y_i \in I_i'\bl \boundary s_i$. (Recall that $\boundary s_i$ denotes the neighborhood of $s_i$.) Let 
\begin{equation}
\label{eq:family-def}
M=\{\{x_i,y_i\}:1\leq i\leq m\}.
\end{equation}
Let $E^{(1)},\dots,E^{(l)}$ be all the subsets of $M$; hence, $l=2^m$.
We let 
$$\mathcal{G}_{m,t}  = \{G^{(j)}=(V_{m,t},E_{m,t}\cup E^{(j)}):1\leq j\leq l\}.$$ 
The graphs in $\mathcal{G}_{m,t}$ satisfy the following.
\begin{fact}
	\label{fact:graph:family}
	If $G=(V,E) \in \mathcal{G}_{m,t}$, then $|V| = (2q+2t-1)m$ and the maximum degree of $G$ is at most
	$$
	\max \big\{q+t, q+\big\lceil (q-1)/t\, \big\rceil +1\big\}. 
	$$
\end{fact}	
\noindent
Using Lemma \ref{lemma:lb_coloring} and ideas similar to those
in the proof of Theorem \ref{thm:ident-lower-bound}, we can show that both structure and equivalent-structure learning are computationally hard in $\mathcal{G}_{m,t}$ (sample complexity is exponential in $m$). 
This immediately implies that \textit{structure learning} is also hard in $\mathcal{G}(n,d)$ 
provided $d$ is large enough so that $\mathcal{G}_{m,t} \subseteq \mathcal{G}(n,d)$.
However, 
this does not necessarily imply
that \textit{equivalent-structure learning} is hard for $\mathcal{G}(n,d)$ which is our goal.
For this, we introduce instead a larger graph family $\mathcal{F}_{m,t}$ that contains $\mathcal{G}_{m,t}$.
Suppose $d$ is an integer such that
$$d\geq \max \big\{q+t, q+\big\lceil (q-1)/t\, \big\rceil +1\big\},$$ the maximum degree of any graph in $\mathcal{G}_{m,t}$; see Fact~\ref{fact:graph:family}. The graph family $\mathcal{F}_{m,t}$ contains
all the graphs in $\mathcal{G}(n,d)$ that have the same set of colorings as some graph in $\mathcal{G}_{m,t}$. Namely, 
$$\mathcal{F}_{m,t}=\{G\in\mathcal{G}(n,d):\OO_{G}=\OO_{G'},G'\in\mathcal{G}_{m,t}\},$$
and $\mathcal{G}_{m,t}\subseteq \mathcal{F}_{m,t}\subseteq \mathcal{G}(n,d)$ by definition. 

The next theorem establishes a lower bound for any equivalent-structure learning algorithm for $\mathcal{F}_{m,t}$. 
We will later see that this lower bound applies also to the graph family $\mathcal{G}(n,d)$, which establishes Part \ref{part:col-exp} of Theorem \ref{thm:col}.

\begin{theorem}\label{thm:lb_coloring_n}
	Let $q,m\in\mathbb{N}^+$, $q\geq 3$ and $t=\lceil \sqrt{q-1}\, \rceil$. Any equivalent-structure learning algorithm 
	for proper vertex $q$-colorings and the graph family $\mathcal{F}_{m,t}$ 
	that succeeds with probability at least $q\cdot\exp\big[-m/(2(q-1))\big]$ requires at least 
	$\exp\big[m/(2(q-1))\big]$ samples.
\end{theorem}

\noindent
The following generalization of Fact \ref{fact:lb_eta} will be used in the proof of Theorem~\ref{thm:lb_coloring_n}.

\begin{fact}
	\label{fact:lb_eta:general}
	Let $H$ be an arbitrary constraint graph and
	let $\mathcal{F}_1,\dots,\mathcal{F}_r$ be $r$ families of distinct $H$-colorable graphs.
	Suppose that for all $1 \le i \le r$ every graph in the in $\mathcal{F}_i$ has the same set of $H$-colorings $\Omega_{\mathcal{F}_i}$.
	Assume also that $\Omega_{\mathcal{F}_i} \neq \Omega_{\mathcal{F}_j}$ for $i \neq j$ and that $\Omega_{\mathcal{F}_r} \subseteq \Omega_{\mathcal{F}_i}  \subseteq \Omega_{\mathcal{F}_1}$ for all $1 \le i \le r$.
	Let
	$\mathcal{\hat{F}}=\mathcal{F}_1 \cup \dots \cup \mathcal{F}_r$ and let
	\begin{equation*}
	\eta = 1-\frac{|\OO_{\mathcal{F}_r}|}{|\OO_{\mathcal{F}_1}|}.
	\end{equation*}
	If there exists an equivalent-structure learning algorithm for $H$ and $\mathcal{\hat{F}}$, such that for any $G^*\in\hat{\mathcal{F}}$, given $L$ independent samples from $\pi_{G^*}^H$ as input, it outputs a graph $G$ satisfying $\OO_{G}=\OO_{G^*}$ with probability at least $1/r+\aA$ with $\alpha > 0$, then $L\geq \aA/\eta$. 
\end{fact}

\noindent
Note that Fact \ref{fact:lb_eta} corresponds to the special case where each family $\mathcal{F}_i$ contains a single graph, and thus it follows immediately from Fact \ref{fact:lb_eta:general}.
We are now ready to prove Theorem~\ref{thm:lb_coloring_n}.

\newenvironment{myproof-thm:lb_coloring_n}{\paragraph{\textbf{Proof of Theorem \ref{thm:lb_coloring_n}}}}{\hfill$\square$}
\begin{myproof-thm:lb_coloring_n}
	Let $G^{(1)}=G_{m,t}$ and $G^{(l)}=G_{m,t}\cup M$ where $M$ is as in (\ref{eq:family-def}) and $l=2^m$. Let $\mathcal{F}_j$ be the class of graphs that contains all the graphs in $\mathcal{F}_{m,t}$ that has the same set of colorings as $G^{(j)}$. (Recall that $\mathcal{F}_{m,t}$ is the set of graphs in $\mathcal{G}(n,d)$ that have the same set of colorings as some graph in $\mathcal{G}_{m,t}$.) Let $\OO_{\mathcal{F}_j}=\OO_{G^{(j)}}$ for all $j$. Note that $\OO_{\mathcal{F}_l} \subseteq \OO_{\mathcal{F}_j} \subseteq \OO_{\mathcal{F}_1}$ for all $1\leq j\leq l$.
	Let
	$$\eta = 1 - \frac{|\OO_{\mathcal{F}_l}|}{|\OO_{\mathcal{F}_1}|}.$$
	We establish a lower bound $\eta$ and then apply Fact~\ref{fact:lb_eta:general} to prove the theorem.
	
	By Lemma \ref{lemma:lb_coloring} each $q$-coloring of $G_{m,t}$ (i.e., of $G^{(1)}$) is determined by the colors of the independent sets $I_1,I_1'$ and of the vertices $s_1,\ldots,s_m$. Then, 
	the number of $q$-colorings of $G^{(1)}$ where all vertices in $I_1$ and $I_1'$ receive the same color is equal to $q(q-1)^m[(q-1)!]^{2m}$, since there are $q$ choices for the color of $I_1$ and $I_1'$, $q-1$ choices for the color of each $s_i$, and $(q-1)!$ colorings for each $C_i$ and $C_i'$. Similarly, the number of colorings where $I_1$ and $I_1'$ receive distinct colors is equal to $q(q-1)(q-2)^m[(q-1)!]^{2m}$. Thus, the probability that in a uniform random $q$-coloring of $G^{(1)}$ the vertices in $I_1$ and $I_1'$ have the same color is 
	\begin{align*}
	\frac{q(q-1)^m[(q-1)!]^{2m}}{q(q-1)^m[(q-1)!]^{2m}+q(q-1)(q-2)^m[(q-1)!]^{2m}} 
	&= 1 - \frac{(q-2)^m}{(q-1)^{m-1}+(q-2)^m} \\
	&\geq 1 - (q-1)\Big(\frac{q-2}{q-1}\Big)^{m} \\
	&\geq 1 - (q-1)\e^{-\frac{m}{q-1}}.
	\end{align*}
	
	Let $\sS$ be a $q$-coloring of $G^{(1)}$ where $I_1$ and $I_1'$ receive the same color. 
	Then, Lemma~\ref{lemma:lb_coloring} implies that  $\sS(x_i) \neq \sS(y_i)$ for all $1\leq i\leq m$, since $C_i$ and $I_i'$ have distinct colors. 
	(Recall that $x_i \in C_i$ and $y_i \in I_i'\bl \partial s_i$ are the vertices used to define the graph family $\mathcal{G}_{m,t}$).
	Therefore, $\sS$ is a proper $q$-coloring of $G^{(l)}$; hence $\sS\in\OO_{\mathcal{F}_l}$ and
	\begin{equation*}
	\frac{|\OO_{\mathcal{F}_l}|}{|\OO_{\mathcal{F}_1}|} = {\Pr}_{\pi_{\mathcal{F}_1}}[\sigma \in \OO_{\mathcal{F}_l}] \geq {\Pr}_{\pi_{\mathcal{F}_1}} [\sS(I_1)=\sS(I_1')] \geq 1 - (q-1)\e^{-\frac{m}{q-1}}.
	\end{equation*}
	(Note that $\pi_{\mathcal{F}_1}=\pi_{G^{(1)}}=\pi_{G_{m,t}}$.) Then,
	\begin{equation}
	\label{eq:eta}
	\eta = 1-\frac{|\OO_{\mathcal{F}_l}|}{|\OO_{\mathcal{F}_1}|} \leq (q-1){\e}^{-\frac{m}{q-1}}.
	\end{equation}
	
	Every graph in $\mathcal{G}_{m,t}$ is a distinct supergraph of $G_{m,t}$ and $|\mathcal{G}_{m,t}| = 2^m$.
	Moreover, for any $G = (V,E) \in \mathcal{G}_{m,t}$ 
	and any $1\le i \le m$
	such that $\{x_i,y_i\} \not\in E$,
	there are $q$-colorings of $G$ where $x_i$ and $y_i$ are assigned the same color.
	Consequently, for any  $G, G' \in \mathcal{G}_{m,t}$, we have $\OO_G\neq \OO_{G'}$ whenever $G\neq G'$. 
	Then, $\OO_{\mathcal{F}_i} \neq \OO_{\mathcal{F}_j}$ for any $i\neq j$, and by definition all the graphs in $\mathcal{F}_i$ are distinct for each $i$. 
	Therefore, Fact~\ref{fact:lb_eta:general} implies that to equivalently learn any $G\in\mathcal{F}_{m,t}$ with probability at least $2^{-m}+\aA$, the number of random samples needed is $L \geq \aA/\eta$.
	Setting
	\begin{equation*}
	\aA=q\e^{-\frac{m}{2(q-1)}}-2^{-m} > 0,
	\end{equation*}
	we get that to equivalently learn a graph  $G\in\mathcal{F}_{m,t}$ with success probability at least $q\e^{-\frac{m}{2(q-1)}}$, we require
	\begin{equation*}
	L \geq \frac{q\e^{-\frac{m}{2(q-1)}}-2^{-m}}{(q-1)\e^{-\frac{m}{q-1}}} \geq \frac{q\e^{-\frac{m}{2(q-1)}}-\e^{-\frac{m}{2(q-1)}}}{(q-1)\e^{-\frac{m}{q-1}}} \ge \e^{\frac{m}{2(q-1)}}
	\end{equation*}
	where the second inequality follows from $2\geq \e^{\frac{1}{2(q-1)}}$ when $q\geq 3$.
\end{myproof-thm:lb_coloring_n}
\vspace{1em}

\noindent
The following corollary of Theorem \ref{thm:lb_coloring_n} corresponds to Part \ref{part:col-exp} of Theorem \ref{thm:col}.

\begin{corollary}
	\label{cor:intro-strong-bound}
	Let $q,n,d\in\mathbb{N}^+$ such that $3 \le q < d - \sqrt{d}+\Theta(1)$ and $ n \ge 2q+2\lceil \sqrt{q-1}\, \rceil - 1$. 
	Then, there exists a constant $c>0$ such that
	any equivalent-structure learning algorithm 
	for $q$-colorings 
	and the graph family $\mathcal{G}(n,d)$ that
	succeeds with probability at least $\exp(-cn)$ 
	requires at least $\exp(cn)$ samples.
\end{corollary}

\begin{proof}
	Let $k = 2q+2\lceil \sqrt{q-1}\, \rceil -1$. If $k$ divides $n$, then take 
	$m = n/k$.
	By Fact \ref{fact:graph:family},
	every graph in $\mathcal{G}_{m,t}$ has $n=mk$ vertices
	and maximum degree $q+\lceil \sqrt{q-1}\, \rceil + 1$ and so
	$\mathcal{G}_{m,t} \subseteq \mathcal{F}_{m,t} \subseteq \mathcal{G}(n,d)$ provided $d \ge q+\lceil \sqrt{q-1}\, \rceil + 1$.
	Theorem~\ref{thm:lb_coloring_n} implies that
	there exists $c = c(q) > 0$
	such that
	any equivalent-structure learning algorithm for $\mathcal{F}_{m,t}$
	with success probability at least $\exp(-cn)$ 
	requires $\exp(cn)$ samples.
	By definition, the set of $q$-colorings of any graph in $\mathcal{G}(n,d)\backslash \mathcal{F}_{m,t}$ is distinct from the set of $q$-colorings of any graph in $\mathcal{F}_{m,t}$. Since also $\mathcal{F}_{m,t}\subseteq \mathcal{G}(n,d)$, equivalent-structure learning in $\mathcal{G}(n,d)$ is strictly harder than in $\mathcal{F}_{m,t}$. 
	Specially, any equivalent-structure learning algorithm for $ \mathcal{G}(n,d)$
	with success probability at least $\exp(-cn)$ 
	requires at least $\exp(cn)$ samples. 
	Note that $d \ge q+\lceil \sqrt{q-1}\, \rceil + 1$ implies
	$q < d - \sqrt{d} + \Theta(1)$.
	
	The result follows in similar fashion when
	$k$ does not divide $n$, but we are required to modify slightly the graph families $\mathcal{G}_{m,t}$ and $\mathcal{F}_{m,t}$.
	Suppose $n = km + r$ where $1 \le  r \le k-1$ and let $W = \{w_0,\dots,w_{r-1}\}$ be a simple path.
	For every $G \in \mathcal{G}_{m,t}$, 
	add $W$ and the edge $\edge{s_m}{w_0}$ to $G$ to obtain a graph $\hat{G}$.
	Let $\mathcal{\hat{G}}_{m,t}$ be the resulting graph family.
	Every graph in $\mathcal{\hat{G}}_{m,t}$ has exactly $n$ vertices and maximum degree $q+\lceil \sqrt{q-1}\, \rceil + 1$.
	Moreover, every $q$-coloring of $G \in \mathcal{G}_{m,t}$ corresponds to exactly $(q-1)^r$ colorings of $\hat{G} \in \mathcal{\hat{G}}_{m,t}$.
	Define $\mathcal{\hat{F}}_{m,t}$ as before; i.e., $\mathcal{\hat{F}}_{m,t}$ is the set of all graphs
	in $\mathcal{G}(n,d)$ that have the same set of colorings as some graph in $\mathcal{\hat{G}}_{m,t}$.
	The argument in the proof of Theorem \ref{thm:lb_coloring_n} 
	and Fact \ref{fact:lb_eta:general}
	imply that 
	any equivalent-structure learning algorithm for $\mathcal{\hat{F}}_{m,t}$
	with success probability at least $\exp(-cn)$ 
	requires $\exp(cn)$ independent samples, where $c=c(q) > 0$ is a suitable constant.
	Since $\mathcal{\hat{G}}_{m,t} \subseteq \mathcal{\hat{F}}_{m,t} \subseteq \mathcal{G}(n,d)$ for $d \ge q+\lceil \sqrt{q-1}\, \rceil + 1$, the result follows.
\end{proof} 
\noindent
We conclude this section with the proofs of Lemma~\ref{lemma:lb_coloring}, Fact~\ref{fact:graph:family} and Fact~\ref{fact:lb_eta:general}.

\newenvironment{myproof-lemma:lb_coloring}{\paragraph{\textbf{Proof of Lemma \ref{lemma:lb_coloring}}}}{\hfill$\square$}
\begin{myproof-lemma:lb_coloring}
	Let $\sigma$ be a $q$-coloring of $G_{m,t}$.
	For $1\leq i< m$, since every vertex in $I_i$ is adjacent to every vertex in $C_i$, all vertices in $I_i$ have the same color in $\sigma$, which is the one color not used by $C_i$. 
	Moreover, every vertex in $C_{i+1}$ is adjacent to a vertex of $I_i$, and so $C_{i+1}$ is colored with the same set of $q-1$ colors as $C_i$ in $\sigma$. Then, every clique $C_1,\ldots,C_m$ is colored with the same set of $q-1$ colors and every independent set $I_1,\ldots,I_m$ is colored with the one remaining color in $\sigma$.
	The same holds for $C_1',\ldots,C_m'$ and $I_1',\ldots,I_m'$ as well.
\end{myproof-lemma:lb_coloring}
\vspace{0.2em}

\newenvironment{myproof-fact:graph:family}{\paragraph{\textbf{Proof of Fact \ref{fact:graph:family}}}}{}
\begin{myproof-fact:graph:family}
	The number of vertices in $G_{m,t}$ is $m(|C_1|+|C_1'|+|I_1|+|I_1'|+1)=(2q+2t-1)m$. The degree of the vertices in the cliques $C_i$ or $C_i'$ is at most $q-2+t+1=q-1+t$. 
	Moreover, the degree of the vertices in the independent sets $I_i$ or $I_i'$ is at most $q-1+\lceil (q-1)/t \rceil +1=q+\lceil (q-1)/t \rceil$. Thus, the maximum degree of $G_{m,t}$ is no more than
	\begin{equation*}
	\max \big\{q+t-1, q+\big\lceil (q-1)/t\, \big\rceil \big\}.
	\end{equation*}
	Therefore, the maximum degree of any graph in $\mathcal{G}_{m,t}$ is at most
	\begin{equation*}
	\max \big\{q+t, q+\big\lceil (q-1)/t\, \big\rceil +1\big\}. \tag*{$\square$}
	\end{equation*}
\end{myproof-fact:graph:family}
\vspace{-1.3em}

\newenvironment{myproof-fact:lb_eta:general}{\paragraph{\textbf{Proof of Fact \ref{fact:lb_eta:general}}}}{\hfill$\square$}
\begin{myproof-fact:lb_eta:general}
	Let $\mathcal{A}$ be any (possibly randomized) equivalent-structure learning algorithm that, given $L$ independent samples $\Gamma = (\sS^{(1)},\ldots,\sS^{(L)}) \in \OO_{\mathcal{F}_i}^L$ from an unknown distribution $\pi_{\mathcal{F}_i}=\pi_{G^*}$ for some $G^*\in\mathcal{F}_i$, outputs a graph $\mathcal{A}(\Gamma)$ in $\mathcal{\hat{F}}$. 
	For any $G^*$, the probability that $\mathcal{A}$ equivalently learns the graph given $L$ independent samples from $\pi_{\mathcal{F}_i}$ is
	\begin{align*}
	\Pr [ \mathcal{A}(\Gamma) \in \mathcal{F}_i ]
	={}\sum_{x \in \OO_{\mathcal{F}_i}^L} {\Pr} [\Gamma=x]
	\Pr[\mathcal{A}(x) \in \mathcal{F}_i].
	\end{align*}
	Recall that by assumption $\OO_{\mathcal{F}_r}\subseteq \OO_{\mathcal{F}_i}$. Let $\mathcal{T}$ be the set of all sample sequences $\sS^{(1)},\ldots,\sS^{(L)}$
	such that $\sS^{(j)} \not\in \OO_{\mathcal{F}_r}$ for at least one $j$; namely, $\mathcal{T} = \OO_{\mathcal{F}_i}^L \bl \OO_{\mathcal{F}_r}^L$. 
	Note that 
	$$|\mathcal{T}| = |\OO_{\mathcal{F}_i}|^L - |\OO_{\mathcal{F}_r}|^L.$$
	Then,
	\begin{align*}
	{\Pr} [ A(\Gamma) \in \mathcal{F}_i]
	=& \sum_{x \in \OO_{\mathcal{F}_r}^L} \frac{1}{|\OO_{\mathcal{F}_i}|^L} \cdot \Pr[\mathcal{A}(x) \in \mathcal{F}_i] + \sum_{x \in \mathcal{T}} \frac{1}{|\OO_{\mathcal{F}_i}|^L} \cdot \Pr[\mathcal{A}(x) \in \mathcal{F}_i]\\
	\leq& \frac{1}{|\OO_{\mathcal{F}_r}|^L} \sum_{x \in \OO_{\mathcal{F}_r}^L} \Pr[\mathcal{A}(x) \in \mathcal{F}_i] +  \frac{|\OO_{\mathcal{F}_i}|^L - |\OO_{\mathcal{F}_r}|^L}{|\OO_{\mathcal{F}_i}|^L}.
	\end{align*}
	Since $\OO_{\mathcal{F}_i} \subseteq \OO_{\mathcal{F}_1}$, we get
	\begin{equation*}
	\frac{|\OO_{\mathcal{F}_i}|^L - |\OO_{\mathcal{F}_r}|^L}{|\OO_{\mathcal{F}_i}|^L} = 1 - \frac{|\OO_{\mathcal{F}_r}|^L}{|\OO_{\mathcal{F}_i}|^L} \leq 1 - \frac{|\OO_{\mathcal{F}_r}|^L}{|\OO_{\mathcal{F}_1}|^L} = 1-(1-\eta)^L \leq L\eta. 
	\end{equation*}
	Suppose the equivalent-structure learning algorithm $\mathcal{A}$ has success probability at least $1/r+\alpha$, then 
	\begin{equation*}
	{\Pr} [ \mathcal{A}(\Gamma) \in \mathcal{F}_i] \geq \frac{1}{r}+\aA.
	\end{equation*}
	Hence,
	\begin{equation*}
	\frac1r+\aA \leq \frac{1}{|\OO_{\mathcal{F}_r}|^L} \sum_{x \in \OO_{\mathcal{F}_r}^L} \Pr[\mathcal{A}(x) \in \mathcal{F}_i] + L\eta.
	\end{equation*}
	Since $\sum_{i=1}^r \pr[\mathcal{A}(x) \in \mathcal{F}_i] = 1$ for any fixed sample sequence $x$, summing up over $i$ we get
	\begin{gather*}
	1+r\aA \leq \frac{1}{|\OO_{\mathcal{F}_r}|^L} \sum_{x \in \OO_{\mathcal{F}_r}^L} \sum_{i=1}^r \Pr[\mathcal{A}(x) \in \mathcal{F}_i] + rL\eta = 1+rL\eta.
	\end{gather*}
	Hence, $L\geq\aA/\eta$ as claimed.
\end{myproof-fact:lb_eta:general}
\vspace{1em}

\noindent
Combining the results from Sections \ref{subsection:efficient-coloring}, \ref{subsection:identifiability} and \ref{subsection:strong-lower}, we obtain the proof of Theorem~\ref{thm:col}: Part \ref{part:1} follows from Lemmas \ref{lemma:colorings:main} and \ref{lemma:hard-contraints:main}; Part~\ref{part:2} from Theorem \ref{thm:id-col}; and Corollary \ref{cor:intro-strong-bound} implies Part~\ref{part:col-exp} of the theorem.

\subsection{General lower bound for structure learning of $q$-colorings}\label{sec:wlb}
When $d-\sqrt{d} + \Theta(1) \le q \le d$
structure learning for $q$-colorings is not identifiable, 
and the strong lower bound from Part~\ref{part:col-exp} of Theorem \ref{thm:id-col} (i.e., Corollary \ref{cor:intro-strong-bound}) does not apply either. In this subsection we establish a weaker but more general lower bound for proper colorings that applies in this regime. Specifically, we provide a family of graphs $\mathcal{F} \subseteq \mathcal{G}(n,d)$ such that the number of random $q$-colorings required to learn any graph in $\mathcal{F} $ with success probability at least $1/2$ is $\exp(\Omega(d-q))$.

\begin{theorem}
	\label{thm:intro-weak-bound}
	Let $d,q,n \in\mathbb{N}^+$ such that $3 \le q<d$ and $n \ge d+2$. Then,
	any equivalent-structure learning algorithm for $\mathcal{G}(n,d)$ with success probability at least $1/2$
	requires at least $\exp(\Omega(d-q))$ samples.  
\end{theorem}

\noindent
Let $d,q,n \in\mathbb{N}^+$ such that $3 \le q<d$ and $n \ge d+2$. 
Let $C$ be a clique of size $q-3$, let $I$ be an independent set of size $d-q+1$ and
let $W = \{w_1,\dots,w_{n-d-1}\}$ be a simple path. 
Also, let $u,v,w$ be three additional vertices that are not in $C$, $I$ or $W$. 
Define the graph $G=(V,E)$ such that 
$$V= V(C) \cup V(I)\cup  \{w_1,\dots,w_{n-d-1}\} \cup \{u,v,w\},$$
where $V(C)$ and $V(I)$ are the vertices in $C$ and $I$, respectively. 
In addition to the edges in $C$ and $W$, $G$ has the following edges:
\begin{enumerate}
	\item every vertex in $C$ is adjacent to every vertex in $I$;
	\item $u$ and $v$ are adjacent to every vertex in $C$ and $I$;
	\item $w,w_1$ are adjacent to $v$; 
\end{enumerate}
see Figure \ref{fig:d-q}.

\begin{figure}[t]
	\centering
	\begin{tikzpicture}[scale=0.65]
	\node(u) at (-2,0) {$u$};
	\node(v) at (2,0) {$v$};
	\node(w1) at (4,1) {$w_1$};
	\node(w2) at (5.5,1) {$w_2$};
	\node(w3) at (7,1) {$w_3$};
	\node(w4) at (9.3,1) {};
	\node(w) at (4,-1) {$w$};
	\node(C)[circle, draw=black] at (0,2) {$C$};
	\node(I)[shape=circle, draw=black] at (0,-2) {$I$};
	
	\path
	(u) edge (C)
	(u) edge (I)
	(v) edge (C)
	(v) edge (I)
	(C) edge (I)
	(v) edge (w1)
	(v) edge (w)
	(w1) edge (w2)
	(w2) edge (w3)
	(w3) edge[white] node[black]{$\cdots\cdots$} (w4);
	\end{tikzpicture}
	\caption{The graph $G$ where $C=K_{q-3}$ and $I$ is an independent set of size $d-q+1$.}
	\label{fig:d-q}
\end{figure}
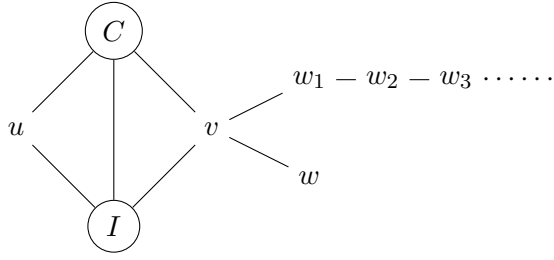

Let 
$$\mathcal{G} = \{G_1=G,G_2=(V,E\cup\{uw\}),G_3=(V,E\cup\{uw_1\}),G_4=(V,E\cup\{uw,uw_1\})\}.$$
Note that every graph in $\mathcal{G}$ is an $n$-vertex graph of maximum degree at most $d$ and so $\mathcal{G} \subseteq \mathcal{G}(n,d)$. 
Furthermore, for $1\leq i\leq 4$ let $\mathcal{F}_i$ be the family of all graphs in $\mathcal{G}(n,d)$ that have the same set of $q$-colorings as $G_i$ and let $\mathcal{F} = \bigcup_{i=1}^4\mathcal{F}_i$. 
The following theorem immediately implies Theorem~\ref{thm:intro-weak-bound}.

\begin{theorem}
	\label{thm:lb_coloring_d-q}
	Let $d,q,n \in\mathbb{N}^+$ such that $3 \le q<d$ and $n \ge d+2$. Then
	the number of independent random $q$-colorings required to learn any graph in $\mathcal{F}$  with probability at least $1/2$ is $\exp(\Omega(d-q))$.
\end{theorem}

\begin{proof}
	If $u$ and $v$ receive distinct colors in a $q$-coloring of $G$, then the clique $C$ will be colored by $q-3$ of the $q-2$ colors not used by $u$ and $v$ and the number of available colors for every vertex in $I$ is only $1$. Thus, the number of colorings of $G$ where $u$ and $v$ receive distinct colors is 
	$q! (q-1)^{n-d}$. Otherwise, 
	if $u$ and $v$ receive the same color in a coloring of $G$, then $C$ will use $q-3$ of the $q-1$ available colors and every vertex in $I$ has $2$ available color choices. Hence, the number of such colorings is $q!(q-1)^{n-d} \cdot 2^{d-q}$. Any $q$-coloring of $G$ where $u$ and $v$ receive the same color is also a proper $q$-coloring of $G_4$.
	Therefore, we get
	\begin{equation*}
	\frac{|\OO_{\mathcal{F}_4}|}{|\OO_{\mathcal{F}_1}|} = \frac{|\OO_{G_4}|}{|\OO_G|} \geq \frac{q!(q-1)^{n-d} \cdot 2^{d-q}}{q!(q-1)^{n-d} + q!(q-1)^{n-d} \cdot 2^{d-q}} = 1-\e^{-\OO(d-q)}.
	\end{equation*}
	Thus, it follows from Fact~\ref{fact:lb_eta:general} that 
	any equivalent-structure learning algorithm for $\mathcal{F}$ that succeeds with probability at least $1/2$ requires
	\begin{equation*}
	L \geq \Big(\frac{1}{2}-\frac{1}{4}\Big)\Big(1-\frac{|\OO_{\mathcal{F}_4}|}{|\OO_{\mathcal{F}_1}|}\Big)^{-1} = \e^{\OO(d-q)}
	\end{equation*}
	samples.
\end{proof}
\noindent
Theorem \ref{thm:intro-weak-bound} follows immediately from Theorem \ref{thm:lb_coloring_d-q} and the fact that $\mathcal{F} \subseteq \mathcal{G}(n,d)$. 

\section{Learning $H$-colorings in Dobrushin uniqueness} 
\label{sec:Dobrushin}

As mentioned in the introduction,
our results in Section \ref{sec:colorings} for statistical identifiability, structure learning and equivalent-structure learning for proper colorings
reveal a tight connection between the computational hardness of these problems
and the uniqueness/non-uniqueness phase transition. 

In this section we explore this connection in a more general setting. 
For this we define the \textit{Dobrushin uniqueness condition},
which is a standard tool in statistical physics for establishing uniqueness of the Gibbs distribution in infinite graphs.

\begin{definition}
	\label{def:dob}
	Let $H$ be an arbitrary constraint graph and let $G=(V,E)$ be an $H$-colorable graph.
	For $w\in V$, let 
	\[   S_w := \{(\tau,\tau_w): \tau,\tau_w\in \{1,\dots,q\}^{|V|}~\textrm{and}~\tau(z)=\tau_w(z) ~\forall z\neq w \}.
	\] 
	For $v,w\in V$, let  
	\[  R_{vw} :=  \max_{(\tau,\tau_w) \in S_w}
	\|\pi_v(\cdot \mid \tau(\partial v)) - \pi_v(\cdot \mid \tau_w(\partial v))\|_{\textsc{tv}},
	\]
	where $\pi_v(\cdot| \tau(\partial v))$ and $\pi_v(\cdot| \tau_w(\partial v))$ 
	are the conditional distributions at $v$ 
	given the respective assignments $\tau$ and $\tau_w$ 
	on the neighbors of $v$.
	Let $  \alpha := \max_{v \in V} \sum_{w \in \partial v } R_{vw}$.
	When $\alpha < 1$,
	$\pi_G$ is said to satisfy the Dobrushin uniqueness condition.
\end{definition}
\noindent
We note that the Dobrushin uniqueness condition (typically) concerns soft-constraint systems on infinite graphs.
The definition we use here for hard-constraint models in finite graphs appeared in \cite{BD}; see also \cite{SS}.

The Dobrushin uniqueness condition implies the following key property. 

\begin{lemma}
	\label{lemma:dob:main}
	Let $H \neq K_q^+$ be an arbitrary constraint graph and
	suppose $\edge{i}{j} \not\in E(H)$. 
	Let $G=(V,E)$ be a graph such that
	$\pi_G$ satisfies the Dobrushin uniqueness condition.
	Then, for all $\edge{u}{v} \not\in E$
	\[\Pr[X_u=i,X_v = j] \ge \frac{(1-\alpha)^2}{q^{2}}.\]
\end{lemma}

\noindent
Lemmas \ref{lemma:dob:main}
and~\ref{lemma:hard-contraints:main} imply
that the 
\textsc{structlearn-H} algorithm
requires $L \!=\! O (q^2 \log (\frac{n^2}{\varepsilon}))$ independent samples 
to succeed with probability at least $1-\varepsilon$
and has running time is~$O(L n^2)$.
This establishes
Theorem~\ref{thm:dob} from the introduction. 

\newenvironment{myproof-lemma:dob:main}{\paragraph{\textbf{Proof of Lemma \ref{lemma:dob:main}}}}{}
\begin{myproof-lemma:dob:main}
	For $u \in V$ we show first that
	$\Pr[X_u = i] \ge (1 - \alpha)/q$. 
	If $\sigma$ is an $H$-coloring of $G$ sampled according to $\pi_G$,
	we may update the color of any vertex $w \in V$ by
	choosing a new color for $w$ uniformly at random
	among the available colors for $w$ given $\sigma(V \setminus w)$.
	The resulting $H$-coloring after this update has distribution $\pi_G$.
	
	Suppose $\sigma_0$ is an $H$-coloring of $G$ sampled according to $\pi_G$,
	and let $\tau_0$ be the color assignment that agrees with $\sigma_0$ everywhere except possibly at $u$, where we set $\tau_0(u) = i$.
	(Note that $\tau$ is not necessarily a valid $H$-coloring.)
	
	Let $\partial u = \{v_1,\dots,v_l\}$. We update the configuration in $v_1$, then in $v_2$ and so on,
	in both $\sigma_0$ and $\tau_0$; then we update the color of $u$.
	Let $\sigma_k$ and $\tau_k$ be the configuration after updating $v_k$ in $\sigma_{k-1}$ and $\tau_{k-1}$, respectively. 
	The color of $v_k$ in both $\sigma_{k-1}$ and $\tau_{k-1}$ is updated 
	using the optimal coupling $\nu_k$ between the distributions  	
	$\pi_{v_k}(\cdot | \sigma_{k-1})$ and 
	$\pi_{v_k}(\cdot | \tau_{k-1})$ as follows.
	Sample $(a_k,b_k)$ from $\nu_k$
	and let $\sigma_k (V\setminus v_k) = \sigma_{k-1} (V\setminus v_k)$, $\sigma_k (v_k) = a_k$,  $\tau_k (V\setminus v_k) = \tau_{k-1} (V\setminus v_k)$ and $\tau(v_k) = b_k$.
	After updating $\partial u = \{v_1,\dots,v_l\}$ in this manner, $\sigma_l$ has law $\pi_G$. Moreover,
	\begin{align*}
	\Pr[\sigma_l \neq \tau_l] 
	&\le {\Pr}_{\nu_l}[\sigma_l \neq \tau_l | \sigma_{l-1} = \tau_{l-1}] + \Pr[\sigma_{l-1} \neq \tau_{l-1}] \\
	&\le \sum_{k=1}^l {\Pr}_{\nu_k}[\sigma_k \neq \tau_k | \sigma_{k-1} = \tau_{k-1}]\\
	&= \sum_{k=1}^l {\|\pi_{v_k}(\cdot |\tau_{k-1})-\pi_{v_k}(\cdot |\tau_{k-1})\|}_{\textsc{tv}} \\
	&\le \alpha,
	\end{align*}
	where the last inequality follows from the definition of the Dobrushin condition. 
	Hence, with probability at least $1-\alpha$,
	$\sigma_l = \tau_l$.
	If this is the case, then color $i$ is compatible with $\sigma_l(V \setminus u)$
	and thus when $u$ is updated 
	it receives color $i$ with probability at least $1/q$.
	Thus, we get
	$$
	\Pr[X_u = i] \ge \frac{(1-\alpha)}{q}.
	$$
	
	Finally, let $v \in V$ such that $v \not\in \partial u$.
	Using the procedure described above to update the configuration in $\partial u \cup u$, and
	then in $\partial v \cup v$ 
	we obtain
	\begin{equation*}
	\Pr[X_u = i,X_v=j] \ge \frac{(1-\alpha)^2}{q^2}.
	\end{equation*}
\end{myproof-lemma:dob:main}

\vspace{-1.12cm}
\hfill$\square$

\section{Approximate-structure learning of $H$-colorings}
\label{section:approx}

	In addition to structure learning (exact recovery of the hidden graph $G$) and equivalent-structure learning (learning a graph with the same set of $H$-colorings), we may consider the corresponding approximation problem of finding a graph $\hat{G}$ such that $\pi_{\hat{G}}$ is close to $\pi_G$ in some notion of distance, such as total variation distance or Kullback-Leibler divergence.
	Apparently, this task is much simpler for hard-constraint systems. 
	

    In this section we consider this approximation variant of structure learning for hard-constraint systems with respect to total variation distance. 
	That is, given $L$ independent samples $\sigma^{(1)},\dots,\sigma^{(L)}$ from $\pi_G$,  we consider the problem
	of finding a graph $\hat{G}$ 
	such that 
	$$
	\TV{\pi_G-\pi_{\hat{G}}} < \gamma,
	$$ 
	where $\gamma>0$ is a desired precision. 

\begin{theorem}
	\label{thm:approx}
	Suppose $H\neq K_q^+$ and let $\hat{G}$ be the output of the \textsc{structlearn-H} algorithm. Then, for all $\varepsilon\in(0,1)$ and $\gamma\in(0,1)$, 
	$$ \Pr\left[\TV{\pi_G-\pi_{\hat{G}}} < \gamma\right] \geq 1-\varepsilon $$ 
	provided $L\geq 4\gamma^{-1}n^2\log(\frac{n^2}{2\varepsilon})$.
\end{theorem}

\noindent
Recall that the running time of \textsc{structlearn-H} is $O(Ln^2)$, so from Theorem \ref{thm:approx} we get
an algorithm for approximate structure learning with running time $O(\gamma^{-1}n^4\log(\frac{n}{\varepsilon}))$.

\newenvironment{myproof-thm:approx}{\paragraph{\textbf{Proof of Theorem \ref{thm:approx}}}}{\hfill$\square$}
\begin{myproof-thm:approx}
	Let $\hat{G} = (V(\hat{G}),E(\hat{G}))$.
	Recall that $\{u,v\}\notin E(\hat{G})$ if and only if $u,v$ receive incompatible colors in one of the samples 
	$\sigma^{(1)},\dots,\sigma^{(L)}$ from 
	$\pi_G$. Hence, $\hat{G}$ is a supergraph of $G$ and so $\Omega_{\hat{G}} \subseteq \Omega_{G}$. Moreover, 
	\begin{equation*}
	\TV{\pi_G-\pi_{\hat{G}}} = \sum_{\sigma\in \Omega_{G}\backslash \Omega_{\hat{G}}} \frac{1}{|\Omega_{G}|} = \frac{|\Omega_{G}|-|\Omega_{\hat{G}}|}{|\Omega_{G}|} = \Pr\big[\sigma \notin \Omega_{\hat{G}}\big],
	\end{equation*}
	assuming $\sigma$ is an $H$-coloring of $G$ chosen uniformly at random (i.e., $\sigma$ is drawn from $\pi_G$). If we let $\Gamma = E(\hat{G})\backslash E(G)$, then
		\begin{align}\label{eq:ub}
	 \Pr\big[\sigma \notin \Omega_{\hat{G}}\big]
	= \Pr\Big[ \exists \{u,v\}\in \Gamma: \{\sigma_u,\sigma_v\}\notin E(H) \Big]
	\leq \sum_{\{u,v\}\in \Gamma} \Pr\Big[ \{\sigma_u,\sigma_v\}\notin E(H) \Big]
	\end{align}	
	by a union bound.
	 
	Now, for $\gamma>0$ let 
	$$M_\gamma =\left\{ \{u,v\} \not\in E(G) :\Pr\Big[ \{\sigma_u,\sigma_v\}\notin E(H) \Big] \geq \frac{2\gamma}{n^2} \right\}. $$ 
	Let $Z_{uv}$ be the number of samples $\sigma^{(1)},\dots,\sigma^{(L)}$ where vertices $u$ and $v$ receive incompatible colors. A Chernoff bound implies that, for any $\{u,v\}\in M_\gamma$,
	$$ 
	\Pr[Z_{uv}=0] \leq \Pr\left[Z_{uv} \le \frac{\gamma L}{n^2}\right] \le  \exp\left(\frac{-\gamma L}{4n^2}\right) \leq \frac{2\varepsilon}{n^2}.
	$$
	A union bound then implies that with probability at least $1-\varepsilon$ all edges in $M_\gamma$ are not in $E(\hat{G})$. Hence, with probability at least $1-\varepsilon$, all edges in $E(\hat{G})$ satisfy:
	$$ 
	\Pr\Big[ \{\sigma_u,\sigma_v\}\notin E(H) \Big] < \frac{2\gamma}{n^2}.
	$$
	Plugging this bound into (\ref{eq:ub}), we get
	$$
	\Pr\left[\TV{\pi_G-\pi_{\hat{G}}} < \gamma\right] \geq 1-\varepsilon 
	$$
	as desired.
\end{myproof-thm:approx}

\section{Learning weighted $H$-colorings}
\label{sec:weighted}

In this section we consider the more general setting of weighted $H$-colorings.
We restrict our attention to constraint graphs with at least one hard constraint,  which corresponds to
spin systems
with hard constraints.

\subsection{Spin systems with hard constraints}	

Let $G = (V,E,\tT)$ be an undirected weighted graph with weights given by the function $\theta: E \cup V \rightarrow \R^+$. (For definiteness we only consider a positive weight function $\tT$.)
A $\textit{spin system}$ on the graph $G$ consists of a set of \textit{spins} $[q]=\{1,\dots,q\}$, a symmetric \textit{edge potential} $J: [q] \times [q] \rightarrow \R \cup \{-\infty\}$ and a \textit{vertex potential} $h: [q] \rightarrow \R$. A \textit{configuration} $\sigma:V\to [q]$ of the system is an assignment of spins to the vertices of $G$. Each configuration $\sigma \in [q]^V$ is assigned probability
\begin{equation}\label{pdf}
\pi_G (\sigma) = \frac{1}{Z_G} \exp\left(\sum_{(u,v) \in E} \theta(u,v)J(\sS_u,\sS_v)\;  + \;\sum_{u \in V} \tT(u) h(\sS_u)\right),
\end{equation}
where $Z_G$ is the normalizing constant called the \textit{partition function}. If $J(i, j) = -\infty$ for some $i, j \in [q]$, then $\{i, j\}$ is a hard constraint; otherwise $i$ and $j$ are compatible. 

Unweighted $H$-colorings, which were considered in Sections \ref{sec:identify-unweight}, \ref{sec:colorings} and \ref{sec:Dobrushin}, correspond to the special case
where $\theta=1, h=0$ and
\[  J(i,j) = \begin{cases}  
1 & \mbox{ if } (i,j)\in E(H); \\
-\infty & \mbox{ if } (i,j)\notin E(H).  
\end{cases}
\]

In this section we consider the structure learning problem for a class of models known as {\em permissive systems}. This is a widely used
notion in statistical physics for spin systems with hard constraints; see, e.g.,~\cite{DSVW,MSW,DMS}.
There are several different notions in the literature, but we consider here the weakest one (i.e., the easiest to
satisfy). Roughly, the condition says that for any boundary condition there is always a valid configuration
for the interior.  

\begin{definition}
	A spin system is called \textit{permissive} if for any $A \subseteq V$ and any valid configuration $\tau$ on $V \setminus A$, there is at least one valid configuration $\sS$ on $A$ such that $\pi (\sS | \tau) > 0$.
\end{definition}


\noindent
Independent sets, and more generally the hard-core model, are examples of permissive models
since we can assign spin $0$ (unoccupied) to 
the vertices in $A$.

\subsection{Structure learning for spin systems with hard constraints}

We first formalize the notion of structure learning for the setting of weighted 
constraint graphs.
Suppose we know the number of spins $q$, the edge potential $J\in \R^{q\times q}$ and the vertex potential $h\in \R^q$ of a spin system $\mathcal{S}$. 
Consider the family of graphs
\begin{align*}
\mathcal{G}(n,d,\aA,\bB,\gG) = \{ G=(V,E,\tT): {}&|V|=n,\\
&\DD(G) \leq d,\\
&\aA \leq |\tT(u,v)| \leq \bB \text{ for all } \edge{u}{v}\in E,\\
&|\tT(v)| \leq \gG \text{ for all } v\in V \},
\end{align*}
where $\DD(G)$ denotes the maximum degree of the graph $G$.
Suppose that
we are given $L$ independent samples $\sigma^{(1)},\sigma^{(2)},\ldots,\sigma^{(L)}$ from the distribution $\pi_G$ where $G \in \mathcal{G}$. 
A \textit{structure learning algorithm} for the spin system $\mathcal{S}$ and the family $\mathcal{G}(n,d,\aA,\bB,\gG)$ 
takes as input the sample sequence $\sigma^{(1)},\sigma^{(2)},\ldots,\sigma^{(L)}$
and outputs an estimator $\hat{G} \in \mathcal{G}(n,d,\aA,\bB,\gG)$ such that $\Pr[G=\hat{G}] \ge 1-\varepsilon$, where $\varepsilon > 0$ is a prescribed failure probability. 

\subsection{Learning permissive spin systems}\label{sec:permissive}

In this section we analyze the running time and sample complexity of the \textsc{structlearn-H} algorithm for permissive spin systems.

Let 
$\hat{\gamma} = \gamma \cdot \max_{i \in [q]} |h(i)|$ 
and
$\hat{\beta} =  \beta \cdot \max_{i,j \in [q]}\, |J(i,j)|$.
Recall that for $v \in V$,
$X_v$ denotes the random variable for the color of $v$ under $\pi_G$. We show that for permissive systems, the running time of \textsc{structlearn-H} is polynomial in the size of the graph, but depends exponentially on $\hat{\gamma}$, $\hat{\beta}$ and its maximum degree.

\begin{theorem}
	\label{thm:permissive-wt}
	Let $G=(V,E) \in \mathcal{G}(n,d,\aA,\bB,\gG)$ and 
	suppose that $\mathcal{S}$ is a permissive spin system with at least one hard constraint.
	Then, if the structure learning algorithm
	receives as input 
	$$L \ge 8 q^{2(d+1)}{\e}^{4(2\hat{\beta} d^2 + \hat{\gamma})} \log \left(\frac{n^2}{2\varepsilon}\right)$$ independent samples from $\pi_G$,
	it outputs 
	the graph $G$
	with probability at least $1-\varepsilon$
	and has running time $O(L n^2)$.
\end{theorem}
\noindent
Theorem~\ref{thm:permissive-wt} yields a structure learning algorithm for the hard-core model for all~$\lambda>0$; thus it generalizes the algorithmic result of~\cite{BGS-hc}. We observe also that the running time of our algorithm for permissive systems is comparable
to the running time of the optimal structure learning algorithms for soft-constraint systems in \cite{KM}.

Theorem~\ref{thm:permissive-wt} is a direct corollary of the following lemma and Lemma~\ref{lemma:hard-contraints:main}.
\begin{lemma}
	\label{lemma:permissive:main}
	Suppose that $\mathcal{S}$ is a permissive spin system with at least one hard constraint $\edge{i}{j} \in [q] \times [q]$ on a graph $G=(V,E) \in \mathcal{G}(n,d,\aA,\bB,\gG)$. Then, for all $\edge{u}{v} \not\in E$
	\[\Pr[X_u=i,X_v = j] \ge \frac{1}{q^{2(d+1)}{\e}^{4(2\hat{\beta} d^2 + \hat{\gamma})}}.\]
\end{lemma}				
\noindent
In the proof of Lemma \ref{lemma:permissive:main} we use the following fact.

\begin{fact}
	\label{fact:permissive:min_bounds}
	Let $R \subseteq V$ and
	let $\tau$ be a configuration on $\boundary R$.
	If $\Omega^\tau(R) \neq \emptyset$ is the set of valid configurations on $R$ given $\tau$, then for any $\sigma \in \Omega^\tau(R)$
	\[\Pr[X_R = \sigma \mid X_{\boundary R}=\tau] \ge \frac{1}{q^{|R|}{\e}^{2(\hat{\beta} d|R| + \hat{\gamma})}}.\]
\end{fact}
\noindent
We are now ready to prove Lemma \ref{lemma:permissive:main}.

\newenvironment{myproof-lemma:permissive:main}{\paragraph{\textbf{Proof of Lemma \ref{lemma:permissive:main}}}}{}
\begin{myproof-lemma:permissive:main}
	For any $A \subseteq V$ and any spin configuration $\sigma$ of $A$,	with a slight abuse of notation we use $\{\sigma\}$
	for the event $\{X_{A} = \sigma\}$.
	
	Let $u,v \in V$ such that $\edge{u}{v} \not\in E$ and
	let $N_1$ and $N_2$ be the set of vertices at distances one and two, respectively, from $\{u,v\}$; i.e., $N_1 = \boundary u \cup \boundary v$ and $N_2 = \{w \in \boundary N_1: w \neq u, w\neq v\}$.
	Let $\Omega_1$ and $\Omega_2$ be the set of valid configurations for $N_1$ and $N_2$, respectively. Then,
	\begin{equation}
	\label{eq:lower:ini}
	\Pr[X_u=i,X_v = j] \ge \min_{\tau_2 \in \Omega_2} \Pr[X_u=i,X_v = j \mid \tau_2]. 
	\end{equation}
	Since the spin system is permissive, for any $\tau_2 \in \Omega_2$ there exists $\tau_1 \in \Omega_1$ such that 
	$$\Pr[\tau_1 \mid \tau_2, X_u = i, X_v = j] > 0.$$
	Then, 
	\begin{equation}
	\label{eq:lower:first}
	\Pr[X_u=i,X_v = j \mid \tau_2]
	\ge \Pr[X_u=i,X_v = j \mid \tau_1] \Pr[\tau_1 \mid \tau_2] 
	\ge \frac{1}{q^2 {\e}^{2(2\hat{\beta} d + \hat{\gamma})}}  \Pr[\tau_1 \mid \tau_2],
	\end{equation}
	by Fact \ref{fact:permissive:min_bounds}.
	Now, 
	\[
	\Pr[\tau_1 \mid \tau_2] 
	= \sum_{a, b\in [q]} \Pr[\tau_1 \mid X_u=a,X_v=b ,\tau_2] \Pr[X_u=a, X_v=b \mid \tau_2].
	\]
	Since $|N_1| \le 2d$, by Fact \ref{fact:permissive:min_bounds},
	$\Pr[\tau_1 \mid X_u=a,X_v=b ,\tau_2] \ge \frac{1}{q^{2d}{\e}^{2(2\hat{\beta} d^2 + \hat{\gamma})}}$.
	Together with (\ref{eq:lower:ini}) and (\ref{eq:lower:first}) this implies 
	\[\Pr[X_u=i,X_v = j ] \ge \frac{1}{q^{2(d+1)}{\e}^{4(2\hat{\beta} d^2 + \hat{\gamma})}}.\]
\end{myproof-lemma:permissive:main}	

\vspace{-1.12cm}
\hfill$\square$
\vspace{0.3em}

\begin{remark}
	A simplified version of this argument can be used to show that in a permissive $H$-coloring, 
	for any hard constraint $\edge{i}{j} \not\in E(H)$,
	$\Pr[X_u = i,X_v = j] \ge 1/q^{2d}$. From this we obtain a structure learning algorithm
	for permissive $H$-colorings with running time $O(q^{2d} n^2 \log n)$ via
	Lemma~\ref{lemma:hard-contraints:main}.
\end{remark}
\noindent
We conclude this section with the proof of Fact \ref{fact:permissive:min_bounds}.

\newenvironment{myproof-fact:permissive:min_bounds}{\paragraph{\textbf{Proof of Fact \ref{fact:permissive:min_bounds}}}}{}
\begin{myproof-fact:permissive:min_bounds}
	For $\sigma \in \Omega^\tau(R)$ let 
	\[w(\sigma) = \exp\left[\sum_{u \in R} \sum_{v \in \boundary u \cap R} \theta(u,v) J(\sigma_u,\sigma_v) +\sum_{u \in R} \sum_{v \in \boundary u \cap \boundary R} \theta(u,v) J(\sigma_u,\tau_v)+ \theta(u) h(\sigma_u) \right].\]
	Then,
	$$\Pr[X_R = \sigma \mid X_{\boundary R}=\tau] = \frac{w(\sigma)}{Z_{R,\tau}},$$
	with $Z_{R,\tau} = \sum_{\sigma' \in \Omega^\tau(R)} w(\sigma')$. Observe that for all $\sigma' \in \Omega^\tau(R)$,
	${\e}^{-\hat{\beta} d|R| - \hat{\gamma} } \le w(\sigma') \le {\e}^{\hat{\beta} d|R| + \hat{\gamma}}$.
	Hence, $Z_{R,\tau} \le q^{|R|} {\e}^{\hat{\beta} d|R| + \hat{\gamma}}$ and
	\[ \Pr[X_R = \sigma \mid X_{\boundary R}=\tau] \ge \frac{1}{q^{|R|} {\e}^{2(\hat{\beta} d|R| + \hat{\gamma})}}.\]
\end{myproof-fact:permissive:min_bounds}

\vspace{-1.12cm}
\hfill$\square$
\vspace{0.3em}

\subsection{Identifiability for weighted $H$-colorings}
\label{sec:identify-weighted}

We prove next an analog of our characterization theorem (Theorem \ref{thm:identify}) for identifiability
of weighted $H$-colorings. 
The edge potential $J$ 
corresponds to the weighted adjacency matrix of a weighted constraint graph $H^{J}=(V(H^{J}),E(H^{J}))$,
where $V(H^{J}) = \{1,\dots,q\}$,
$\{i,j\}\notin E(H^{J})$ iff $J(i,j)= -\infty$,
and the weight of $\{i,j\} \in E(H^{J})$ is $J(i,j)$.
As before, we say that a graph $G$ is 
$H^J$-colorable if there is a valid $H^J$-coloring for $G$.
If $\{i,j\}\notin E(H^{J})$ we call $\{i,j\}$ a hard constraint.  
The notion of identifiability extends
to the weighted setting as follows.

\begin{definition}
	A weighted constraint graph $H^J$ is said to be identifiable with respect to a family of $H^J$-colorable graphs $\mathcal{G}$ if for any two distinct
	graphs
	$G_1,G_2\in\mathcal{G}$ we have $\pi_{G_1} \neq \pi_{G_2}$. In particular, when $\mathcal{G}$ is the set of all finite $H^J$-colorable graphs we say that $H^J$ is identifiable.
\end{definition}
\noindent
(Definition~\ref{def:identify} is the analog definition in the unweighted setting.)

In our characterization theorem
we consider the supergraphs $G_{ij}$'s 
introduced in the unweighted setting; see Definition~\ref{dfn:Hij} and Figure \ref{fig:Hij}.

\begin{theorem}
	\label{thm:weighted-id}
	Let $H^J$ be a weighted constraint graph with at least one hard constraint.
	If $H^J$ has a self-loop, then $H^J$ is identifiable.
	Otherwise $H^J$ is identifiable if and only if for each $\edge{i}{j} \in E(H^J)$ there exists an $H^J$-coloring $\sigma$ of $G_{ij}$ such that
	\begin{equation*}
	J(\sS_i,\sS_j) + J(\sS_{i'},\sS_{j'}) \neq J(\sS_{i'},\sS_j) + J(\sS_i,\sS_{j'}).
	\end{equation*}
\end{theorem}
\noindent
(Recall that $i'$ and $j'$ are the copies of the vertices $i$ and $j$ in $G_{ij}$.)

\vspace{1em}
\begin{proof}
	For clarity, we shall assume in this proof that the underlying graph $G=(V,E)$ is 
	unweighted and that there is no external field; i.e., $\theta=1$ and $h=0$.  
	The same proof generalizes 
	to spin systems on
	weighted graphs with external field.
	
	Henceforth we use $H$ for $H^J$ to simplify the notation.
	We consider first the case when $H$ has no self-loops.
	For the forward direction we consider the contrapositive.
	Suppose that there exists $\{i,j\}\in E(H)$ such that for every proper 
	$H$-coloring $\sigma$ of $G_{ij}$ we have
	\begin{equation*}
	J(\sS_i,\sS_j) + J(\sS_{i'},\sS_{j'}) = J(\sS_{i'},\sS_j) + J(\sS_i,\sS_{j'}).
	\end{equation*}
	Under this assumption
	we construct two distinct $H$-colorable graphs $F_1,F_2$ such that $\pi_{F_1} = \pi_{F_2}$;
	this implies that $H$ is not identifiable, which would complete the proof of the forward direction.
	For this, for each $\edge{i}{j} \in E(H)$ let us
	define the supergraph $G_{ij}^2$ of $H$ that is the result 
	of creating two copies $i',i''$ of vertex $i$ 
	and two copies $j',j''$ of vertex $j$, with no edges between $i',i'',j',j''$. 
	Formally, for each $\edge{i}{j}\in E(H)$, 
	we define the graph $G_{ij}^2 = (V(G_{ij}^2),E(G_{ij}^2))$ as follows:
	\begin{enumerate}
		\item $V(G_{ij})=V(H)\cup\{i',i'',j',j''\}$ where $i',i'',j',j''$ are four new colors;
		\item If $\edge{a}{b} \in E(H)$, then the edge $\edge{a}{b}$ is also in $E(G_{ij}^2)$;
		\item For each $k \in V(G_{ij}^2) \setminus \{i',i'',j',j''\}$, 
		the edges $\edge{i'}{k}$ and $\edge{i''}{k}$ are in $G_{ij}^2$  
		if and only if the edge $\edge{i}{k}$ is in $H$, and 
		similarly $\edge{j'}{k},\edge{j''}{k}\in E(G_{ij}^2)$ if and only if $\edge{j}{k} \in E(H)$; 
	\end{enumerate}
	see Figure~\ref{fig:Gij2} for an example.

	\begin{figure}[t]
		\centering
		\begin{minipage}{0.24\linewidth}
			\centering
			\begin{tikzpicture}
			\node(3) at (-0.75,0) {$3$};
			\node(4) at (0.75,0) {$4$};
			\node(1) at (-0.75,1.5) {$1$};
			\node(2) at (0.75,1.5) {$2$};
			\node(1') at (-0.75,2.5) {};
			\node(2') at (0.75,2.5) {};
			\node(1'') at (-0.75,3.5) {};
			\node(2'') at (0.75,3.5) {};
			
			\draw (1)--(2);
			\draw (1)--(3);
			\draw (2)--(4);
			\draw (3)--(4);
			\end{tikzpicture}
			\caption*{$H$}
		\end{minipage}
		\begin{minipage}{0.24\linewidth}
			\centering
			\begin{tikzpicture}
			\node(3) at (-0.75,0) {$3$};
			\node(4) at (0.75,0) {$4$};
			\node(1) at (-0.75,1.5) {$1$};
			\node(2) at (0.75,1.5) {$2$};
			\node(1') at (-0.75,2.5) {$1'$};
			\node(2') at (0.75,2.5) {$2'$};
			\node(1'') at (-0.75,3.5) {$1''$};
			\node(2'') at (0.75,3.5) {$2''$};
			
			\path
			(1) edge (2)
			(1') edge (2)
			(2') edge (1)
			(1'') edge (2)
			(2'') edge (1)
			(1) edge (3)
			(1') edge [bend right] (3)
			(1'') edge [bend right] (3)
			(2) edge (4)
			(2') edge [bend left] (4)
			(2'') edge [bend left] (4)
			(3) edge (4);
			\end{tikzpicture}
			\caption*{$G_{12}^2$}
		\end{minipage}
		\begin{minipage}{0.24\linewidth}
			\centering
			\begin{tikzpicture}
			\node(3) at (-0.75,0) {$3$};
			\node(4) at (0.75,0) {$4$};
			\node(1) at (-0.75,1.5) {$1$};
			\node(2) at (0.75,1.5) {$2$};
			\node(1') at (-0.75,2.5) {$1'$};
			\node(2') at (0.75,2.5) {$2'$};
			\node(1'') at (-0.75,3.5) {$1''$};
			\node(2'') at (0.75,3.5) {$2''$};
			
			\path
			(1) edge (2)
			(1') edge (2)
			(2') edge (1)
			(1'') edge (2)
			(2'') edge (1)
			(1) edge (3)
			(1') edge [bend right] (3)
			(1'') edge [bend right] (3)
			(2) edge (4)
			(2') edge [bend left] (4)
			(2'') edge [bend left] (4)
			(3) edge (4)
			(1') edge (2')
			(1'') edge (2'');
			\end{tikzpicture}
			\caption*{$F_1$}
		\end{minipage}
		\begin{minipage}{0.24\linewidth}
			\centering
			\begin{tikzpicture}
			\node(3) at (-0.75,0) {$3$};
			\node(4) at (0.75,0) {$4$};
			\node(1) at (-0.75,1.5) {$1$};
			\node(2) at (0.75,1.5) {$2$};
			\node(1') at (-0.75,2.5) {$1'$};
			\node(2') at (0.75,2.5) {$2'$};
			\node(1'') at (-0.75,3.5) {$1''$};
			\node(2'') at (0.75,3.5) {$2''$};
			
			\path
			(1) edge (2)
			(1') edge (2)
			(2') edge (1)
			(1'') edge (2)
			(2'') edge (1)
			(1) edge (3)
			(1') edge [bend right] (3)
			(1'') edge [bend right] (3)
			(2) edge (4)
			(2') edge [bend left] (4)
			(2'') edge [bend left] (4)
			(3) edge (4)
			(1') edge (2'')
			(1'') edge (2');
			\end{tikzpicture}
			\caption*{$F_2$}
		\end{minipage}
		\caption{A constraint graph $H$, its supergraph $G_{12}^2$ and the graphs $F_1$ and $F_2$.}
		\label{fig:Gij2}
	\end{figure}

	Let $\sigma$ be an $H$-coloring of $G_{ij}^2$. Since the subgraphs induced
	by $V(G_{ij}^2)\setminus\{i^*,j^*\}$ with $i^* \in \{i',i''\}$ and $j^*=\{j',j''\}$
	are all isomorphic to $G_{ij}$, our assumption implies
	\begin{align}
	J(\sS_i,\sS_j) + J(\sS_{i'},\sS_{j'}) &= J(\sS_{i'},\sS_j) + J(\sS_i,\sS_{j'}), \label{eq:1}\\
	J(\sS_i,\sS_j) + J(\sS_{i''},\sS_{j''}) &= J(\sS_{i''},\sS_j) + J(\sS_i,\sS_{j''}),\label{eq:2}\\
	J(\sS_i,\sS_j) + J(\sS_{i''},\sS_{j'}) &= J(\sS_{i''},\sS_j) + J(\sS_i,\sS_{j'}), \label{eq:3}\\
	J(\sS_i,\sS_j) + J(\sS_{i'},\sS_{j''}) &= J(\sS_{i'},\sS_j) + J(\sS_i,\sS_{j''}). \label{eq:4}
	\end{align}
	Since the sum
	of the right-hand sides of $(\ref{eq:1})$ and $(\ref{eq:2})$
	is equal to the sum of the right-hand sides of $(\ref{eq:3})$ and $(\ref{eq:4})$, we get
	\begin{equation}
	\label{eq:id:cond}
	J(\sS_{i'},\sS_{j'}) + J(\sS_{i''},\sS_{j''}) = J(\sS_{i''},\sS_{j'}) + J(\sS_{i'},\sS_{j''}).
	\end{equation}
	
	Now, let 
	\begin{align*}
	F_1&= (V(G_{ij}^2),E(G_{ij}^2) \cup \{\edge{i'}{j'},\edge{i''}{j''}\}), \\
	F_2&= (V(G_{ij}^2),E(G_{ij}^2) \cup \{\edge{i''}{j'},\edge{i'}{j''}\});
	\end{align*}
	see Figure~\ref{fig:Gij2}. 
	Then, using $(\ref{pdf})$, for any $H$-coloring $\sigma$ of $G_{ij}^2$ we have
	\begin{equation*}
	\frac{\pi_{F_1}(\sigma)}{\pi_{F_2}(\sigma)} = \frac{Z^{-1}_{F_1}\exp\big(J(\sS_{i'},\sS_{j'}) + J(\sS_{i''},\sS_{j''})\big)}{Z^{-1}_{F_2}\exp\big(J(\sS_{i''},\sS_{j'}) + J(\sS_{i'},\sS_{j''})\big)} = \frac{Z_{F_2}}{Z_{F_1}},
	\end{equation*}
	which is a constant independent of $\sigma$. 
	By (\ref{eq:id:cond}), $\pi_{F_1}$ and $\pi_{F_2}$ have the same support.
	Moreover, any
	$H$-coloring of $F_1$ or $F_2$ is also an $H$-coloring of $G_{ij}^2$. Hence, we conclude that $\pi_{F_1}=\pi_{F_2}$, implying $H$ is not identifiable. 
	This completes the proof of the forward direction.
	
	For the reverse direction suppose that for all $\edge{i}{j} \in E(H)$ there exists an $H$-coloring $\sigma$ of $G_{ij}$ 
	where 
	$$J(\sS_i,\sS_j) + J(\sS_{i'},\sS_{j'}) \neq J(\sS_{i'},\sS_j) + J(\sS_i,\sS_{j'}).$$
	
	Consider two $H$-colorable graphs $G_1=(V,E_1)$ and $G_2=(V,E_2)$ such that $\pi_{G_1}=\pi_{G_2}$. 
	We show that for any $u,v \in V$, $\edge{u}{v} \in E_1$ iff $\edge{u}{v} \in E_2$ and thus $G_1 = G_2$. This implies that $H$ is identifiable.
	
	Let $u,v \in V$ and
	let $\tau$ be an $H$-coloring of $G_1$.
	Since $\pi_{G_1}=\pi_{G_2}$, then $\tau$ is also an $H$-coloring of $G_2$.
	Suppose $i=\tau(u)$ and $j=\tau(v)$. If $i$ and $j$ are not compatible, 
	then $\{u,v\}\not\in E_1$ and $\{u,v\}\not\in E_2$. 
	Thus, let us assume $i$ and $j$ are compatible. 
	Let
	$\sigma$ be an $H$-coloring of $G_{ij}$ such that  
	$$J(a,b) + J(a',b') \neq J(a',b) + J(a,b'),$$
	where $a=\sS_i,b=\sS_j,a'=\sS_{i'},b'=\sS_{j'}$; we know such an $H$-coloring exists by assumption.
	
	Now consider 
	the conditional distribution $\pi_{G_1,uv}^\tau$ on the vertices $u$ and $v$ of $G_1$ given the configuration $\tau(V\setminus\{u,v\})$.
	Then,
	\begin{align*}
	p_1(G_1) &:= \pi_{G_1,uv}^\tau(X_u=a,X_v=b) = \frac{1}{Z_{\mathrm{cond}}(G_1)} \exp[h_{a}+h_{b}+\1(\{u,v\}\in E_1) J(a,b)]\\
	p_2(G_1) &:= \pi_{G_1,uv}^\tau(X_u=a',X_v=b') = \frac{1}{Z_{\mathrm{cond}}(G_1)} \exp[h_{a'}+h_{b'}+\1(\{u,v\}\in E_1)J(a',b')]\\
	p_3(G_1) &:= \pi_{G_1,uv}^\tau(X_u=a',X_v=b) = \frac{1}{Z_{\mathrm{cond}}(G_1)} \exp[h_{a'}+h_{b}+\1(\{u,v\}\in E_1)J(a',b)]\\
	p_4(G_1) &:= \pi_{G_1,uv}^\tau(X_u=a,X_v=b') = \frac{1}{Z_{\mathrm{cond}}(G_1)} \exp[h_{a}+h_{b'}+\1(\{u,v\}\in E_1)J(a,b')],
	\end{align*}
	where $Z_{\mathrm{cond}}(G_1)$ is the normalizing factor for $\pi_{G_1,uv}^\tau$,
	\begin{align*}
	h_{a} &= \sum_{w \in \boundary u} J(a,\tau(w)), \\
	h_{b} &= \sum_{w \in \boundary v} J(b,\tau(w))	
	\end{align*}
	and
	$h_{a'}, h_{b'}$ are defined in similar manner.
	This gives 
	$$
	\frac{p_1(G_1)p_2(G_1)}{p_3(G_1)p_4(G_1)} = \exp\left[\1(\{u,v\}\in E_1)(J(a,b)+J(a',b')-J(a',b)-J(a,b'))\right]
	$$
	Since by assumption $J(a,b)+J(a',b')-J(a',b)-J(a,b') \neq 0$,
	$\{u,v\} \in E(G_1)$ if and only if $p_1(G_1)p_2(G_1) \neq p_3(G_1)p_4(G_1)$.
	Moreover,  $\pi_{G_1}=\pi_{G_2}$ and thus $p_k(G_1)=p_k(G_2)$ for $k \in \{1,2,3,4\}$.
	Hence, 
	$$
	\frac{p_1(G_1)p_2(G_1)}{p_3(G_1)p_4(G_1)} = \frac{p_1(G_2)p_2(G_2)}{p_3(G_2)p_4(G_2)}.
	$$
	This implies that $\{u,v\}\in E_1$ iff $\{u,v\}\in E_2$ and so $G_1=G_2$.
	This completes the proof of the reverse direction when $H$ does not have self-loops.
	When $H$ has at least one self-loop then using the argument in the proof of Lemma \ref{lemma:id:no-loops}, which generalizes straightforwardly to the weighted setting, 
	we get that $H$ is identifiable.
\end{proof}

\bibliographystyle{plain}
\bibliography{Learning}

\end{document}